\documentclass[twocolumn,10pt]{IEEEtran}

\input{def.tex}
\usepackage{dsfont}
\usepackage{minibox}
\DeclareSymbolFont{matha}{OML}{txmi}{m}{it}% txfonts
\DeclareMathSymbol{\varv}{\mathord}{matha}{118}
\usepackage{eurosym}
\usepackage{hhline}
\usepackage{multicol,bbm}

\makeatletter

\IEEEoverridecommandlockouts
\begin{document}
\title{SDN-enabled MIMO Heterogeneous Cooperative Networks with Flexible Cell Association}
\author{Anastasios Papazafeiropoulos, Pandelis Kourtessis, Marco Di Renzo,  John M. Senior, and Symeon Chatzinotas \vspace{2mm} \\
% % $^{*}$Institute for Digital Communications, University of Edinburgh, Edinburgh EH9 3JL, U.K.\\
% % $^{\dag}$SnT - securityandtrust.lu, University of Luxembourg, Luxembourg\\
% % Email: \{a.papazafeiropoulos, t.ratnarajah\}@ed.ac.uk, \{shree.sharma, symeon.chatzinotas\}@uni.lu
\thanks{A. Papazafeiropoulos is with the Communications and Intelligent Systems Research Group,
University of Hertfordshire, Hatfield, U. K. and with SnT at the University of Luxembourg, Luxembourg. P. Kourtessis, and J.M. Senior are with the  Communications and Intelligent Systems Research Group,
University of Hertfordshire, Hatfield, U. K. M. Di Renzo is with the Laboratoire des Signaux et Syst\`emes, CNRS, CentraleSup\'elec, Universit\'e Paris Sud, Universit\'e Paris-Saclay, France.
E-mails: tapapazaf@gmail.com, \{p.kourtessis,j.m.senior\}@herts.ac.uk),marco.direnzo@l2s.centralesupelec.fr, symeon.chatzinotas@uni.lu.
} \thanks{This work was supported in part by FNR, Luxembourg, through the CORE project ECLECTIC.}}
% %%%%%%%%%%%%%%%%%%%%%%%%%%%%%%%%%%%%%%%%%%%%%%%%%%%%%%%%%%%%%%%%%%%%%
% \vspace{-1cm}
\maketitle
\vspace{-12 mm}

\begin{abstract}
 Small-cell densification is a strategy enabling the offloading of users from macro base stations (MBSs), in order to alleviate their load and increase the coverage, especially, for cell-edge users. In parallel, as the network increases in density, the BS cooperation emerges as an efficient design method towards the demands for drastic improvement of the system performance against the detrimental overall interference.  We, therefore, model and scrutinize a heterogeneous network (HetNet) of two tiers (macro and small cells) with multiple-antenna BSs serving a multitude of users, which differ with respect to their basic design parameters, e.g., the deployment density, the number of transmit antennas, and transmit power. In addition, the tiers are enhanced with cell association policies by introducing the concept of the association probability. Above this and motivated by the advantages of cooperation among BSs, the small base stations (SBSs) are enriched with this property in their design. SBS cooperation allows shedding light into its impact on the  cell selection rules in multi-antenna HetNets. Under these settings, software-defined networking (SDN) is introduced smoothly to play the leading role in the orchestration of the network.   {In particular, heavy  operations such as the coordination and the cell association are undertaken by virtue of an SDN controller performing and managing efficiently the corresponding computations  due to its centralized adaptability and dynamicity  towards the enhancement and potential scalability of the network}. In this context, we derive the coverage probability and the mean achievable rate. Not only we show the outperformance of BS cooperation over uncoordinated BSs, but we also demonstrate that the SBS cooperation enables the admittance of more users from the macro-cell BSs (MBSs). Furthermore, we show that by increasing the number of BS antennas, the system performance is improved as the metrics under study reveal. Moreover, we investigate the performance of different transmission techniques, and we identify the optimal bias in each case when SBSs cooperate. Finally, we depict that the SBS densification is beneficial until a specific density value since a further increase does not increase the coverage probability.
\end{abstract}

 
 {\section{Introduction}
The emerging fifth-generation (5G) wireless communication networks aim at an exponential growth in data rates, roughly $1000$ times of the 4G systems~\cite{Andrews2014,Ericsson2015}. In particular, among the promising technologies, the concept of heterogeneous networks (HetNets) concerns low-power nodes with different characteristics, being randomly located according to a Poisson point process (PPP)~\cite{Andrews2013,Kamel2016}. In such networks, macro and small cells, differing in transmit power, spatial density, and coverage, coexist as different tiers. Lots of research in academia and industry has been devoted to this area (see~\cite{Andrews2012} and reference therein), and standardization activities have a long time now been initiated in 3GPP~\cite{Damnjanovic2011,Sankaran2012}. }

 {Although the beginning of the research in HetNets involved single-antenna base stations (BSs)~\cite{Andrews2011,Dhillon2012}, current studies have considered the symbiosis and synergy of HetNets and multi-antenna techniques as the literature reveals~\cite{Dhillon2013,Li2016,DiRenzo2016,Papazafeiropoulos2018}. In this direction, powerful tools from stochastic geometry have allowed the tractable characterization of even multi-user multiple-input multiple-output (MU-MIMO) HetNets~\cite{Dhillon2013,PapazafComLetter2016,Papazafeiropoulos2017,Papazafeiropoulos2018}. One of the main observations is that single user beamforming (SU-BF) may result in better coverage on the downlink than multi-user (MU) beamforming such as space division multiple access (SDMA) in the case of perfect channel state information (CSI). In this direction, the inevitable realistic effect of channel uncertainty has been studied in depth in terms of quantized CSI and imperfect CSI due to pilot contamination, e.g., see~\cite{Kountouris2012} and~\cite{Papazafeiropoulos2017}, respectively. In addition, contributions have been noted by analyzing the impact of hardware impairments and channel aging~\cite{Papazafeiropoulos2017}.}

 {Focusing further on HetNets, we observe that, in practice, load unbalances take place. Among the major sources of this effect is the variation in 
transmit powers. Hence, the coexistence of macrocell BSs (MBSs) and small base stations (SBSs), serving different numbers of users, is unavoidable~\cite{Jo2012,Gupta2014,Andrews2014a,Liu2016}. Actually, in such case, the application of SBSs enables the offloading of users from the MBSs, which results in the improvement of the quality of service (QoS) of the network. Notably, the offloading, described by the association probability, can be achieved by introducing an artificial bias to expand the range of SBSs. The prevailing strategy to obtain the desired bias is the maximization of the coverage. In particular, the authors in~\cite{Jo2012} came to the conclusion that the user association with the BS providing the largest received power to achieve maximization of the signal-to-interference-plus-noise ratio (SINR) is not the right criterion for a multi-antenna HetNet as in HetNets with single-antenna BSs. Instead, the selection rule includes the addition of an appropriate bias regarding the received power to achieve the maximum coverage. }

 {Among the key technologies, addressing the problem of interference is coordinated multi-point transmission (CoMP)~\cite{Simeone2009,Gesbert2010,Irmer2011,AlHaija2017}. In fact, CoMP achieves higher spectral efficiency and coverage of the cell-edge users by exploiting or even mitigating the interference when cooperation among the BSs takes place. Especially, in~\cite{Gesbert2010}, BS cooperation was incorporated with multi-antenna processing principles to combine their benefits and open new research avenues. Moreover, in~\cite{AlHaija2017}, SBSs cooperation was considered during the uplink transmission of a HetNet to derive its spectral efficiency.}

 {5G and Internet of Things (IoT) design works involve complex data management, where the synergism of different technologies is not only indicated but also necessitated. Inevitably, HetNets design, being among the main elements of 5G networks, implicates the interconnection of different interfaces and protocols~\cite{Xia2015a,Haque2016,Bera2017}. The increasing network complexity calls for a shift from a hardware-based approach to a software-based avenue~\cite{Ameigeiras2015}. A promising solution, termed software-defined networking (SDN), has been proposed to cover the arising gap~\cite{Nunes2014,Arslan2015,Sagar2016,Rawat2017}. Its name implies that the network functions are software based. Actually, the functions are manipulated by means of a central controller that decouples the control and data planes, and thus, it controls easier and more efficiently the network, e.g., SDN contributes to higher performance when BS cooperation is implemented. Furthermore,~\cite{Kitindi2017} has considered SDN and centralized radio access network (C-RAN) to realize wireless network virtualization (WNV). Another example is~\cite{Han2016b}, where the authors proposed a traffic load balancing framework striving for a balance between network utilities by means of an SDN establishment\footnote{ {It is worthwhile to mention that another promising technique, dealing with the increasing wireless traffic, is spectrum sharing~\cite{Zhang2017}.}}.}

 {\subsection{Motivation-Central Idea}
This paper is motivated by the following observations:~1) HetNet design provides a more realistic system evaluation, 2) by equipping the BSs with multiple antennas, the capability for application of different transmission techniques with numerous advantages is enabled, 3) BS cooperation brings gains to the spectral efficiency and cell-edge user coverage, 4) it is indicated that HetNets and MIMO coexist and complement each other, 5) macro-cell networks usually need traffic relief and seek SBSs for offloading multiple users, and 6) SDN is able to provide centralized manipulation tasks such as the control of the exchange of the load information and its optimization.}

 {These observations suggest that leveraging the SDN architecture will be advantageous towards a more efficient BS cooperation in an MU-MIMO HetNet enabled with flexible cell association (offloading) policies, i.e., when SBSs are able to reduce the load of an MBS. In addition, SDN offers a platform to enable the SBS cooperation allowing even more users to be offloaded from the MBS. Evidently, this framework presents numerous benefits such as better communication quality, especially for the cell-edge users, and even reduction of the energy consumption since the MBS will consume less energy due to offloading and the whole system will be managed more efficiently by means of SDN. }

 {\subsection{Contributions }
The main contributions are summarized as follows.}
\begin{itemize}
\item  {We present a novel analytical model for the downlink of MU-MIMO HetNet enriched offloading and BS cooperation properties, and orchestrated by an SDN controller.  Contrary to existing works, i.e.,~\cite{Gupta2014} as well as~\cite{Han2016}, we assume BS cooperation and the introduction of SDN as well as MU-MIMO transmission, respectively. Note that the analysis and description are not trivial since the association with a BS in MIMO HetNets demands the maximization of the SINR for coverage maximization, while in single-antenna HetNets  the maximization of the received power is considered sufficient (if no extra bias is introduced)\footnote{In general, MIMO HetNets can exploit the spatial dimension to create non-uniform coverage areas. In such case, a cell association bias can be used to favour SBSs, but in our analysis we assume zero bias for simplicity. }. Especially, the analysis in multi-antenna HetNets concerns complex algebraic manipulations including the Laplace transform of the interference and calculation of its derivative.}
\begin{itemize}
\item   {We derive the probabilities that a user is associated with an MBS, an SBS, and the SBS cluster, respectively. Moreover, in each case, we obtain the probabilistic distance to the tagged BS/BS cluster for a specific bias.}

\item   {We derive the coverage probability and the mean achievable rate of an MU-MIMO HetNet with SBS cooperation and offloading function implementable by an SDN controller. For the sake of comparison, we also present and illustrate the results corresponding to noncooperative BSs.}
\end{itemize}

 \item  {We shed light on the impact of the system parameters and transmission methods on the downlink coverage probability and mean achievable rate, and we make comparisons between the cooperative and noncooperative scenarios as well as between the single and multiple-antenna BSs. Actually, the BS cooperation is beneficiary in all cases. Among the results, we observe that:}
 \begin{itemize}
 \item  {single-user beamforming (SUBF) is preferable with respect to SDMA and single-input single-output (SISO) transmissions even when BSs cooperate because of less interference and the beamforming gain, respectively.}
 \item  {In noncooperation, the coverages yielded by the MBS and the SBS are identical if the various parameters are the same. However, if the SBSs cooperate, the provided coverage is considerably enhanced with respect to the MBS capabilities.}
 \item  {The optimal bias of SBSs cooperation when SDMA and SUBF are implemented moves to the left since the SBS cluster offers higher transmit power. Furthermore, if we increase the optimal bias of SBSs during cooperation we achieve even higher coverage.}
 \item  {Further deployment of BS antennas is beneficial if the number of serving users is kept constant.
 \item Increase of the SBs density results in expanded coverage. Remarkably, an indefinite increase of this density is profitless.}
\end{itemize}
\end{itemize}

The remainder of this paper is structured as follows. Section~\ref{System} presents the basic parameters of the system model of a two-tier HetNet with randomly located BSs  having multiple antennas, and serving multiple users, where SBSs are able to cooperate. In addition, offloading between different classes of BSs is enabled. The section continues with the exposition of the downlink transmission. Section~\ref{main} presents 
the criterion for making the cell selection 
and the association region for each scenario, which defines the area that a typical user is associated with each tier. Next, we present the main results in terms of the coverage probability and mean achievable rate in both noncooperative and cooperative scenarios of MU-MIMO HetNets. The numerical results are placed in Section~\ref{Numerical}, while Section~\ref{Conclusion} summarises the paper.

\textit{Notation:} Vectors and matrices are denoted by boldface lower and upper case symbols. The symbol $(\cdot)^\H$ expresses the Hermitian transpose operator, while the expectation operator is denoted by $\EE\left[\cdot\right]$. The notations $\mathbb{C}^{M \times 1}$ and $\mathbb{C}^{M\times N}$ refer to complex $M$-dimensional vectors and $M\times N$ matrices, respectively. Furthermore, $\mathcal{L}_{I}\!\left(s \right)$ typifies the Laplace transform of $I$. Finally, $\bb \sim \cC\cN{(\b0,\mathbf{\Sigma})}$ represents a circularly symmetric complex Gaussian variable with zero-mean and covariance matrix $\mathbf{\Sigma}$.

\section{System Model}\label{System}
This section presents the formulation of the downlink design and the corresponding SINR of a software-defined HetNet embodying the principles of MU-MIMO transmission. The software-defined features allow the separation of control and data planes to improve the manageability and adaptability of the network. Specifically, we consider two independent tiers corresponding to a macro cell BS network and a network of small cells with BSs being overlaid in different frequency bands\footnote{ {Without loss of generality, the selection of a two-tier network is decided for the sake of exposition of the results extracted from the BS cooperation.}}. In fact, the locations of the BSs in each tier form realizations of the independent homogeneous PPPs $\Phi_{m}$ and $\Phi_{s}$ with densities $\lambda_{m}$ and $\lambda_{s}$, respectively. We use the index $j=\{m,s\}$ to refer by means of $m$ and $s$ to the macro cell and small cell, respectively. Reasonably, both tiers present different other characteristics, i.e., they differ in terms of the transmit power per user $p_{j}$, the number of BS transmit antennas $M_{j}$, the number of users served in each resource block $\Psi_{j}$, the biasing factors, the transmission scheme, and the path-loss exponent $\al_{j}>2$. Moreover, we assume several degrees of freedoms in every cell, which means that the number of BS antennas $M_{j}$ is at least greater than its associated users $ \Psi_{j}$, i.e., $M_{j}\ge \Psi_{j}$\footnote{Without any loss of generality, we assume that the parameters are global in each tier, while different tiers are defined by different parameters. For example, all the SBSs have the same number of antennas $M_{s}$, while $M_{s}\ne M_{m}$.}. Note that the users in both tiers are assumed to be equipped with a single antenna. Obviously, the density of SBSs is higher and they emit lower power. 

\begin{figure}[!h]
 \begin{center}
 \includegraphics[width=0.8\linewidth]{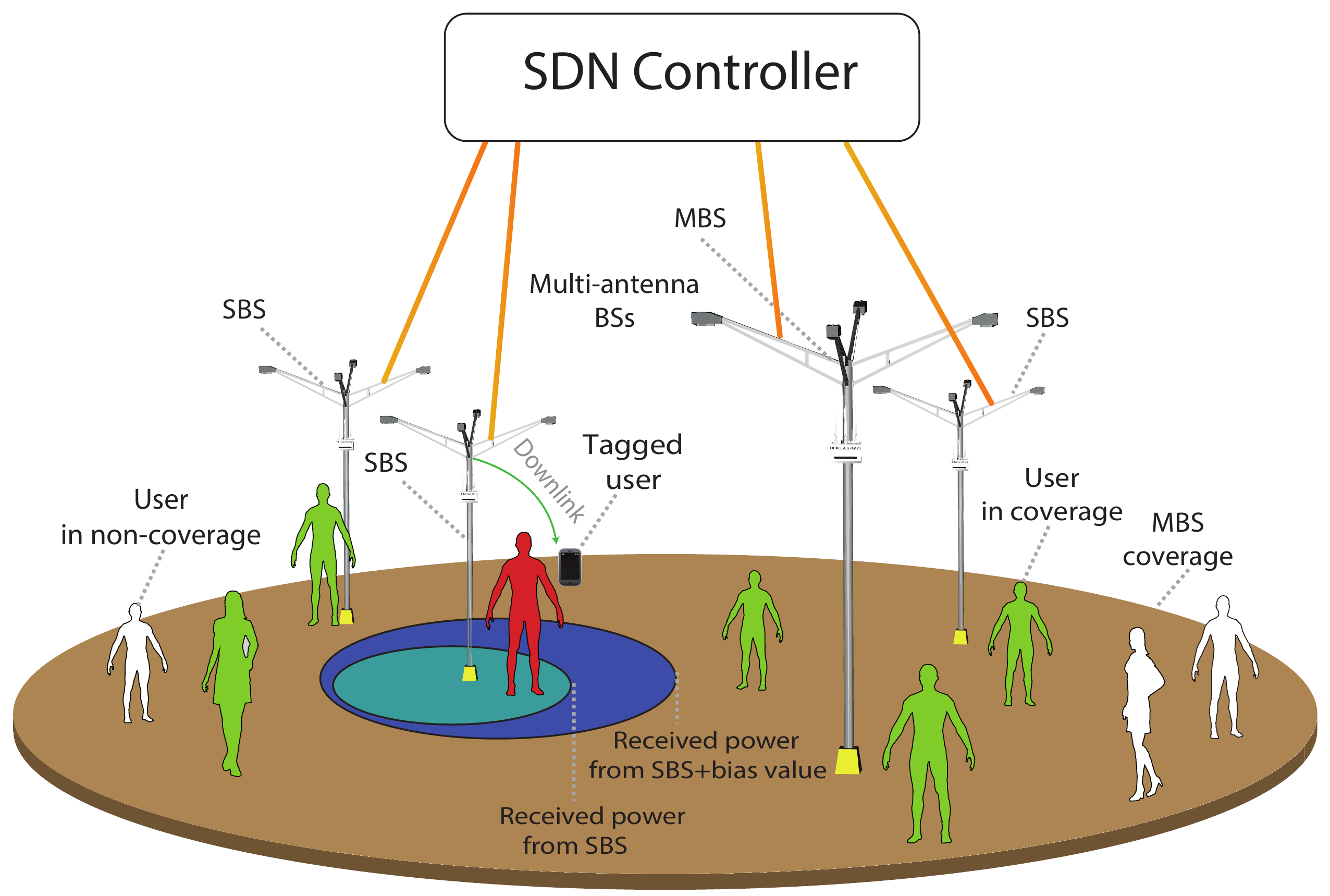}
 \caption{\footnotesize{A two-tier MIMO HetNet, consisted of MBSs and SBSs, and orchestrated by a central SDN controller. All BSs are employed with multiple antennas serving simultaneously several users. SBSs, appearing higher density than the MBSs, can cooperate and are enhanced with cell-association policies. The control of these operations is managed by the SDN controller. The red user represents the typical user, while the green and white users represent customers in coverage and non-coverage, respectively.}}
 \label{scenario}
 \end{center}
 \end{figure}

Based on recent advancements relied on the benefits of SDN~\cite{Kitindi2017}, we employ an SDN controller connected with the BSs of both tiers by means of wired backhaul.  {This controller is programmed to learn the status of each individual network element as well as the physical topology by means of appropriate discovery techniques and databases~\cite{Tanbourgi2014}. As a novel technology, its aims are the reduction of the complexity of 5G networks as well as their more efficient deployment and maintenance. In fact, SDN outperforms non-SDN based networks and its  important role is well-established as shown in~\cite{Qiao2013,Nie2014}. Enabling SDN in HetNets is a promising way and shows great potential for network optimization but certain changes and extensions to the  controller have to be considered~\cite{Arslan2015}. After all, the purpose of applying SDN is its inherited advantageous properties which result in the effective reduction of the exchanged information among the BSs and the reduction of the backhaul power consumption. Notably, the reduction in terms of the exchanged information implies low latency, which is largely desired in future emerging networks.}

  {In this direction, SDoff, being an SBS offloading control mechanism, was proposed in~\cite{Arslan2014}, and it could be considered herein to orchestrate the offloading. In order to standardize the communication between the data plane and control plane, the configuration of all connections relies on the application of the OpenFlow protocol~\cite{Chen2016}. Specifically, the BSs send their state information to the SDN controller, which transmits the control information back to the BSs. Note that this control information includes among others the management of the  cooperation among the SBSs. In particular, we consider measurement flows and control flows similar to~\cite{Li2015a}.}

A significant characteristic of this model is that a user can connect with both tiers. A second meaningful attribute, as it will be shown below, is that a user is able to connect to several SBSs simultaneously. In other words, we assume cooperation among these cells. In fact, the cooperation benefits more when the cells are getting closer, i.e., during the densification of the SBSs that aiming at meeting the increasing traffic demands in 5G networks. The duty-function of the SDN controller is the coordination of the cooperation among the SBSs and the decoupling of the transmission from the processing to attain practically the cooperation.  {As a cooperation model,  we employ joint transmission (JT), which involves the simultaneous  data transmission from multiple coordinated SBSs with appropriate beamforming weights. The choice relies on the fact that JT generally achieves larger performance benefits than  coordinated scheduling and coordinated beamforming (CS/CB), but with larger backhaul overheads which is affordable in our case due to the compensation made by the benefits of the introduced SDN ~\cite{Lee2012}\footnote{SDN is an enabler that facilitates efficient load balancing of the backhaul links. This applies both to data links and collecting CSI. In terms of CSI, both JT and CB schemes require the exchange of CSI among the BSs. In terms of data, the user data have to be transfered to each cooperating BS in JT, while only to a single serving station in CB. Note that SDN ensures that the user bits are directed only to the appropriate BSs without wasting backhaul capacity even in highly dynamic environments where BS cooperation and user associates changes rapidly over time. Actually, in terms of CSI, it can efficiently direct the CSI measurements to the appropriate signal processors, whether they are centralized or distributed.  Hence, JT requires higher capacity backhaul links than CS/CB. This extra burden can be alleviated by SDN.}.}

Basically, we focus on the downlink scenario of communication between a BS and the associated user. Moreover, we assume no intra-cell interference by using orthogonal frequency division multiple access (OFDMA) transmission. However, we consider the interference from other BSs in both tiers.  {Exploiting Slivnyak's theorem, we are able to conduct the analysis by focusing on a typical user, being  a user chosen at random from amongst all users in the network~\cite{Chiu2013a}. Without loss
of generality, we assume that the typical user is located at the origin.} Hence, $x_{i,j}$ and $r_{i,j}$ denote the position and the distance of the $i$th BS in tier $j$ having as reference the typical user. A plausible scenario is shown in Fig.~\ref{scenario}, where a single multi-antenna MBS is surrounded by several SBSs.

\subsection{Downlink Transmission}\label{downlink} 
Herein, we provide a statistical description of the SINR necessitating first to model the downlink transmission. Notably, we assume a general model considering an MU-MIMO architecture, which allows applying a variety of transmission methods from SISO to SDMA. Specifically, we assume that $\bs_{i,j} \in \mathbb{C}^{M_{j} \times 1}$ is the normalized transmit signal vector from the $i$th associated BSs at the $j$th tier to the typical user. The channel vector between the $i$th BS of the $j$th tier and the typical user, located at $r_{i,j}\in \mathbb{R}^{2}$,  is described by $\bff_{i,j} \in \mathbb{C}^{M_{j}\times 1}$. Thus, the received signal at the typical user, found at the $j$th tier, is written as
\begin{align}
 y_{j}&\!=\!\!\sum_{x_{i,j}\in\mathcal{B}}\!\!r_{i,j}^{-\al_{j}/2}\sqrt{p_{j}}\bff_{i,j}^{\H} \bs_{i,j}+\!\!\sum_{x_{i,j}\not\in\mathcal{B}}\!\!r_{i,j}^{-\al_{j}/2}\sqrt{p_{j}}\bff_{i,j}^{\H} \bs_{i,j}+z_{j}\label{signal},
\end{align}
where $ \mathcal{B}$ represents the set of associated BSs with the typical user. Remarkably, this set may include an MBS, an SBS or an SBS cluster.  {In other words, we assume intra-tier cooperation taking place only among the SBSs because we focus on the benefits enjoyed from the implementation of this cooperation. Note that an MBS cooperation is not considered since it has  no practical advantage, and could not take place because between two MBSs normally exist one or more SBSs.} Also, $z_{j}$ is the  {additive white Gaussian noise (AWGN)} with zero mean and power $\sigma_{j}^{2}$. From the physical point of view, the first sum describes the desired part from the associated BS/BSs, while the second sum expresses the interfering part from the rest of the BSs located in both tiers.  {As shown, the MBS and the SBSs share the same spectrum, i.e., the typical user suffers from the interferences coming from other BSs in both tiers.}
\textbf{}
For reasons of tractability and without any loss of generality, we employ in our analysis zero-forcing (ZF) precoding, while more general precoders are left for future work.  {Our analysis enables the application of different transmission methods, e.g., SISO, SUBF, and SDMA. Hence, in Section~V, we perform a comparison among these techniques that exposes the corresponding advantages of each one, and accrediting the optimal strategy under certain parameters.}

 {The deployment of a wired backhaul is generally hard. . In practice, the available CSI in the SDN controller is  imperfect due to inaccurate channel estimation and measurements~\cite{Wang2016 }. Another factor can be the transmission delay (lag) during the CSI delivery.  The current work could be extended to address a comparison between a wireless and a wired backhaul but this is out of the current scope of this work, which is to shed light on the impact of multi-antenna SBS cooperation while the design is facilitated by means of SDN implemented basically like a black box. However, the impact of the links consisting the backhaul is an interesting area left for future work.}

In general, the channel power distribution of a link depends upon its physical representation. For example, it changes if we refer to the desired or the interference part of the received signal, if the BS deploys single or multiple antennas, and the transmission directs to single or multiple users. In the common case of frequency-flat Rayleigh fading with ZF precoding and perfect CSI, the channel power distributions of both direct and interfering links, denoted by $h_{i,j}$ and $g_{i,j}$ follow the Gamma distribution as proved in~\cite{Dhillon2013}\footnote{The practical case of imperfect CSIT can be studied by extending our analysis similar to~\cite{Papazafeiropoulos2017}.}. In detail, we have that $h_{{i,j}}\sim\Gamma (\Delta_{j},1)$ and $g_{i,j}\sim\Gamma (\Psi_{j},1)$, where $\Delta_{j}=M_{j}-\Psi_{j}+1$. The instantaneous received signal power at the typical user from the $i$th BS at the $j$th tier is
\begin{align}
 p_{\mathrm{r}}(x_{i,j})=p_{j}h_{i,j}\|x_{i,j}\|^{-\al_{j}}
\end{align}
with the average received power given by
\begin{align}
 \bar{p}_{\mathrm{r}}(x_{i,j})=p_{j}\Delta_{j}\|x_{i,j}\|^{-\al_{j}}\label{meanPower} .
\end{align}
 {\begin{remark}\label{remarkaveragePower} 
Generally, during the data transmisison phase, the typical user measures the received power or quality of reference signals. According to~\eqref{meanPower}, in the case of SISO, SUBF, and multi-user transmissions, we obtain $\Delta_{j}=1$, $\Delta_{j}=M_{j}$, and $\Delta_{j}=M_{j}-\Psi_{j}+1$. Obviously, the average  received signal power is higher when SUBF is applied, while the MU setting follows. Note that in the special case of SDMA, $\Delta_{j}$ is equal to $1$, i.e., the variable  $\Delta_{j}$ is identical for both transmission techniques. However, further observations regarding the suitable selection of the transmission technique are provided below by investigating the cell-selection rule.
\end{remark}}

\subsection{SINR}
Herein, we present the downlink SINR for transmission at the typical user when different association strategies are applied. Specifically, % \begin{proposition}\label{SINR}
the SINR of the downlink transmission from the associated MBS, SBS or SBS cluster to the typical user is described by
\begin{align}
 \gamma_{j}=\frac{\sum_{x_{i,j}\in\mathcal{B}}p_{j} h_{i,j} r_{i,j}^{-\al_{j}}}{\sum_{x_{i,j}\not \in\mathcal{B}}p_{j} g_{i,j} r_{i,j}^{-\al_{j}}+\sigma_{j}^{2}}.\label{eq:SINR} 
\end{align}

% \end{proposition}
\section{Main Results}\label{main}
This section starts with the presentation of the cell selection rules  the description of the association region. Next, it includes the presentation of the coverage probability and the achievable rate of the typical user in terms of theorems when BS coordination and offloading, achievable in practice by the implementation of SDN, are taken into account. It is noteworthy to add that both MU-MIMO and HetNets result in spatial diversity. Even on this ground, the role of SDN is crucial since it can undertake the management of relevant spatial domain data.

\subsection{Cell Selection}
In a practical HetNet, especially in congested areas, the MBS can be easily heavy loaded with undesired consequences in the performance of the network in terms of quality and energy consumption. Herein, the SBSs can play a significant and beneficial role by alleviating the traffic burden of the MBS. Specifically, SBSs are able to offload users from the congested MBS by applying the cell range expansion technique~\cite{Jo2012,Gupta2014}. Especially, the cell range expansion is an efficient method, where users are offloaded  by means of biasing, in order to increase coverage. In such a case, it is shown that both the coverage and the communication quality are improved, while the macrocell BS reduces its energy consumption  since fewer users need to be served. In addition, SBSs consume lower power in general. Notably, the centralized manipulation by SDN will allow addressing the offloading mechanism. On top of this, network function virtualization (NFV) can be introduced to support the SDN by decoupling the network functions from the hardware~\cite{Li2015}.

Given the need for perpetual connection, we assume open access, where normally a user would be associated with a BS belonging to any tier according to the maximum received-signal-strength criterion. In other words, the user will select as its serving BS the station providing the highest power (the nearest BS) regardless if it is a macro cell or an SBS. However, according to~\cite{Gupta2014}, a practical and general option for cell selection in MU-MIMO HetNets, working optimally in a wide range of systems parameters, is proved to be obtained by means of the maximization of the SINR conditioned on the point process. The maximization can result by adjusting the bias $B_{l}$ values of the so-called biased received power $B_{l} \bar{p}_{\mathrm{r}}(x_{l})$. The corresponding criterion for the selection of a BS from the $l$th tier by the typical user is expressed by means of~\eqref{meanPower} as~\cite{Gupta2014}
\begin{align}
 j=\arg \max_{l\in\{m,s\}} B_{l} \Delta_{l}p_{l}\|x_{i,l}\|^{-\al_{l}},\label{Selection1} 
\end{align}
where $B_{l}=\sqrt{\Psi_{l}/\Delta_{l}}$ is the cell bias selection value\footnote{Note that this selection bias is optimal for coverage maximization, while this does not hold necessarily for other metrics such as the achievable rate.}. Reasonably, the probability that a typical user is associated with a tier depends on the chosen cell association. 
 {\begin{remark}
While selecting the optimal transmission technique, a suitable criterion for the application is not the average received power but~\eqref{Selection1} which depends on the cell bias factor. Herein, the important choice is made based on the product $B_{l} \Delta_{l}$. For instance, among MU-MIMO and SUBF, the former is  preferable since the corresponding product is higher in such case. Furthermore, if SDMA is considered, we have that $B_{l} \Delta_{l}=\sqrt{\Psi_{l}}$, being higher than $1$ provided by SISO transmission. Note that Remark~\ref{remarkaveragePower} verified that~\eqref{meanPower} is not suitable for the cell-selection in MU-MIMO architectures, while~\eqref{Selection1} comes to agreement with known results, e.g.,~\cite{Dhillon2013}.
\end{remark}}

Notably, this work focuses on the study of the association between the typical user and the MBS or the SBS cluster under an MU-MIMO design architecture. Hence, the selection criterion is formulated by the following proposition.
\begin{proposition}
In a multi-antenna SDN HetNet, the typical user selects the $j$th
tier obtained by
\begin{align}
 \!\!j\!=\!\arg \max \!\left( \!\!B_{m} \Delta_{m}p_{m}\|x_{m}\|^{-\al_{m}},\! { \sum_{x_{i,s}\in\mathcal{C}}^{x_{K,s}}\!p_{s}\Delta_{s}r_{i,s}^{-\al_{s}}B_{s}}\!\right)\!\!,\label{Selection} 
\end{align}
where $K$ is the number of SBSs in the cluster.
\end{proposition}
\proof The proof follows similar lines to the selection rule (noncooperative scenario)  given by~\eqref{Selection1} and derived in~\cite{Gupta2014}. Hence, it is omitted to avoid repetition. \endproof
According to this proposition, if $j=m$, the typical user connects with the MBS, otherwise the SBS cluster is selected. Differently to~\cite{Han2016},~\eqref{Selection} describes the more general and practical scenario of multi-antenna BSs serving multiple users.   { Proposition 1 shows that the selection of the SBS cluster is favorable with comparison to the MBS when SBS cooperation takes place in most cases. Of course, this does not reduce the value and importance of the MBS during the network design. The coexistence of an MBS with SBSs is a well-established network architecture and widely accepted in practice by vendors that relies on already known benefits provided across the literature~\cite{Andrews2013,Dhillon2012,Jo2012,Gupta2014,Andrews2014a,Liu2016}. In some cases, according to this Proposition, the received biased power by the MBS might  be higher, and the typical user should be associated with this BS.}
\begin{remark}
Clearly, the modification of the load of each tier, demands the adjustment of the transmit power, the BS density, the number of the cooperative SBSs, the bias value, and the number of BS antennas as well as the number of serving users. Thus, many degrees of freedom are available during the system design. Adjusting the number of antennas per SBS looks a reasonable design in terms of cost and complexity\footnote{Increasing the number of cooperative SBSs is obviously advantageous since the interference turns into a useful signal. However,  {it is not feasible in practice for the users to communicate with many SBSs simultaneously, i.e., the number of cooperative SBS is a design variable that it should take reasonably low values, otherwise, the exchange of the data load among the SBSs can be easily prohibitive. Moreover, before deciding this number,  the corresponding implementation cost has to be taken into consideration and meet the network overall specifications.}}. Also,  the variation in the number of cooperative SBs would be a very efficient solution. Remarkably, simulations show that, in the most cases, the typical user selects the SBS cluster because it provides higher received power.
\end{remark}

% \begin{remark}
% In the case where $\Delta_{k}=\Psi_{k}=1$, i.e., when we account for a single-input single-output (SISO) link, the right part of~\eqref{Selection} reduces to $\bar{p}_{\mathrm{r}}(x_{j})$. Hence, the maximization of the SINR is equivalent to maximizing the average received power. The same situation is faced when spatial division multiple access (SDMA) is applied. If $B_{j}=1/(p_{j }\Delta_{j})$, the user selects the BS providing the smallest path-loss, while if $B_{j}=1/p_{j }$ and the numbers of BS antennas are the same, the selection relies on the cell including the least load. When $B_{j}=1/p_{j }$ and the loads of the cells are the same, the user selects the macrocell BS because, normally, it has more antennas. Finally, when the numbers of BS antennas and users per cell are random, the user associates with the cell offering the largest degrees of freedom ($\max_{j} \Delta_{j}$).\end{remark}

\subsection{Association Region}
The confines, defining the association region of the typical user, are formed based on the selection rule described by~\eqref{Selection}. Actually, since this region encloses the SBS cluster or the macro BS, it constitutes a riddance region for the interfering MBS or the SBS cluster, respectively, which is described by the following proposition.
\begin{proposition}\label{PropRegion} 
When the typical user selects the cooperative cluster, the MBS is found at a distance obeying to the condition
\begin{align}
r_{m}&>\left( \sum_{x_{i,s}\in\mathcal{C}}^{x_{K,s}}\hat{p}_{s}\hat{\Delta}_{s}\hat{B}_{s}r_{i,s}^{-\al_{s}} \right)^{\frac{1}{{\al}_{m}}},\label{prop1} 
\end{align}
where $r_{i,s}$ expresses the distance of the $j$th closest SBS of the cluster to the typical user, $\hat{p}_{s}={p}_{s}/p_{m}$, $\hat{\Delta}_{s}={\Delta}_{s}/\Delta_{m}$, and $\hat{B}_{s}={B}_{s}/B_{m}$.

\end{proposition}
\proof See Appendix~\ref{Alpha}.\endproof

In other words,~\eqref{prop1} describes the area where the MBS is located. Normally, the SBS cluster provides higher biased received power than a single SBS. Thus, the MBS is found further in the cooperative scenario. Henceforth, we make the reasonable assumption that the path-loss in both tiers is identical since the deployment of the SBSs takes place inside the environment of the macro tier.

It is worthwhile to mention that the traffic offloading between the SBS cluster and the MBS can be measured by the associated probability.
\begin{proposition}\label{PropAssoc} 
The probability that the typical user is associated with the SBS cluster is 
\begin{align}
 \mathcal{A}_{\mathrm{SBS_{cl}}}=\int \displaylimits_{\substack{0<r_{1,s}<r_{2,s}\ldots\\r_{K,s}<r_{ \infty,s}}}\mathrm{e}^{ -\lambda_{m}\pi\eta^{2/\al} }f_{\Lambda}\left( \br \right)\mathrm{d}\br, 
\end{align}
where $\eta=1/\left( \hat{p}_{s}\hat{\Delta}_{s} \hat{B}_{s}\displaystyle{\sum_{x_{i,s}\in\mathcal{C}}^{x_{K,s}}}r_{i,s}^{-\hat{\al}} \right)$ and $f_{\Lambda}\left( \br \right)$ is the joint PDF of $r_{i,s}$ with $\br=[r_{1,s},\ldots,r_{K-1,s},r_{K,1}]$ representing the set of distances of the $K$ closest SBSs to the typical user.
\end{proposition}
\proof See Appendix~\ref{AlphaAssoc}.\endproof

\begin{lemma}\cite[Lemma~2]{Han2016}
 The joint PDF of $r_{i,s}^{-\hat{\al}}$ is
 \begin{align}
 f_{\Lambda}\left( \br \right)=\left( 2 \pi \lambda_{s} \right)^{K}\mathrm{e}^{-\lambda_{s}\pi r_{i,s}^{2}} {\prod_{x_{i,s}\in\mathcal{C}}^{x_{K,s}}}r_{i,s}.
 \end{align}
\end{lemma}

In the special case, where the typical user associates with just a single SBS (noncooperative scenario), the corresponding probability is obtained in closed form.
\begin{corollary}\label{corollary1}
When the typical user associates with a single SBS, the probability is given by
\begin{align}
\mathcal{A}_{\mathrm{SBS}}=\frac{1}{1+\frac{1}{\hat{\lambda}_{s}}\frac{1}{\left(\hat{p}_{s} \hat{\Delta}_{s}\hat{B}_{s} \right)\frac{2}{\al}}}\label{corollary2} ,
\end{align}
where $\hat{\lambda}_{s}={\lambda}_{s}/\lambda_{m}$
\end{corollary}
\proof See Appendix~\ref{AlphaAssocCorollary1}.\endproof

\subsection{Coverage Probability}\label{coverage} 
The focal point is the presentation of the downlink coverage probability of the typical user in a MU-MIMO HetNet with the technical derivation given in Appendix~\ref{theoremCoverageProbabilityProof}. However, firstly, we have to provide the PDFs of the distance between a typical user and its serving BS/BSs in both noncooperative and cooperative scenarios. For convenience, in the noncooperative setting, we define the event that the user connects to an MBS or SBS as $A_{m}$ or $A_{s}$, respectively. Similarly, in the case of cooperative architecture, $D_{m}$ or $D_{s}$ corresponds to the event where the user connects to an MBS or the SBS cluster, respectively. 
\begin{proposition}\label{PDFdistance} 
 In the case that the user connects with the MBS, the PDF of the distance becomes
\begin{align}
 f_{R_{s}}\left( r \right)= \frac{2\pi \lambda_{m}r \mathrm{e}^{-\pi r^{2}\left( \lambda_{m}+{\lambda_{s}\beta^{\frac{2}{\al}{}}} \right)}}{1-\mathcal{A}_{\mathrm{SBS}}}.\label{PDFdistanceMBS1} 
\end{align}
The PDF of the distance between the typical user and its associated SBS (noncooperative scenario) is given by
\begin{align}
f_{R_{m}}\left( r \right)= \frac{2\pi \lambda_{s}r \mathrm{e}^{-\pi r^{2}\left( \lambda_{s}+{\lambda_{m}\beta^{-\frac{2}{\al}}}{} \right)} }{\mathcal{A}_{\mathrm{SBS}}},\label{PDFdistanceSBS} 
\end{align}
where $\beta=\left( \hat{p}_{s}\hat{\Delta}_{s}\hat{B}_{s}  \right)^{-1}$ with $\mathcal{A}_{\mathrm{SBS}}$ given by~\eqref{corollary2}. On the other hand (cooperative scenario), the PDF of the distance between the typical user and the associated SBS cluster is written as
\begin{align}
 f_{R_{cs}}\left( \br \right)= \frac{\mathrm{e}^{-\pi \lambda_{m} \eta^{\frac{2}{\al}} }}{\mathcal{A}_{\mathrm{SBS_{cl}}}}f_{\Lambda}\left( \br \right), \label{PDFdistanceSBSCluster} 
\end{align}
where $\mathcal{A}_{\mathrm{SBS_{cl}}}$ is given by~Proposition~\ref{PropAssoc}. In the case of assosiation between the typical user and the MBS, the PDF of the distance is given by
\begin{align}
 f_{R_{cm}}\left( r \right)= \frac{\mathrm{e}^{-\pi \lambda_{m} \eta^{-\frac{2}{\al}} }}{1-\mathcal{A}_{\mathrm{SBS_{cl}}}}f_{r_{m}}\left( r \right),\label{PDFdistanceMBS} 
\end{align}
where $f_{r_{m}}\left( r \right)$ is provided by~\eqref{AlphaAssocCor3}.
\end{proposition}

\proof See Appendix~\ref{PDFdistanceProof}.\endproof

Regarding its definition, the coverage probability for a user associated with a BS in the $j$th tier expresses the probability that the received SINR is greater than a threshold $\mathcal{T}$~ {\cite[Eq.~1]{Andrews2011}}. In mathematical terms, we have
\begin{align}
 P_{\mathrm{c},j}=\mathbb{P}\left( \gamma_{j} >\mathcal{T} \right),\label{defin}
\end{align}
where $\mathbb{P}\left( \cdot \right)$ defines probability.
For the sake of exposition, we define the set of parameters of both tiers as $q=\{\lambda_{j},B_{j},M_{j},\Psi_{j},p_{j}\}$ for $j=\{m,s\}$, and we focus on each scenario separately.
\subsubsection{Noncooperative scenario}
\begin{theorem}\label{theoremCoverageProbability} 
In the noncooperative model, the downlink probability of coverage for a typical user, associated with the $i$th BS in the MBS tier located at a distance $r_{i,m}$ far from the typical user, is given by
\begin{align}
 P_{c,A_{m}}\left( \mathcal{T},q \right)=&\int_{r>0}\sum^{\Delta_{s}}_{k=0}\frac{1}{k!}\left( -s \right)^{k}\frac{\mathrm{d}^{k}}{\mathrm{ds}^{k}}\left( \mathrm{e}^{-sN}\mathcal{L}_{\mathcal{I}}\{s \} \right)\nn\\
 &\times f_{R_{m}}\left( r \right)\mathrm{d}r, \label{theorem1_1}
\end{align}
where $s=p_{m}^{-1}\mathcal{T}r^{\al}$,
while the downlink probability of coverage for a typical user associated with a SBS is provided by
\begin{align}
 P_{c,A_{s}}\left( \mathcal{T},q \right)= &\int_{r>0}\sum^{\Delta_{m}}_{k=0}\frac{1}{k!}\left( -s \right)^{k}\frac{\mathrm{d}^{k}}{\mathrm{ds}^{k}}\left( \mathrm{e}^{-sN}\mathcal{L}_{\mathcal{I}}\{s \} \right)\nn\\
 &\times f_{R_{s}}\left( r \right)\mathrm{d}r,\label{theorem1_2}
\end{align}
where $s=p_{s}^{-1}\mathcal{T}r^{\al}$. The distance distributions are given in Proposition~\ref{PDFdistance}, and the Laplace transform of the interference is obtained by~Proposition~\ref{LaplaceTransform} with its derivative given by~Lemma~\ref{DerivativeOfLaplace}.
\proof See Appendix~\ref{theoremCoverageProbabilityProof}.\endproof

The overall coverage probability of the noncooperative scenario is given by
\begin{align}
P_{c,nc}\!\left( \mathcal{T},q \right) =\left( 1-\mathcal{A}_{\mathrm{SBS}} \right)P_{c,A_{m}}\!\left( \mathcal{T},q \right)+\mathcal{A}_{\mathrm{SBS}} P_{c,A_{s}}\!\left( \mathcal{T},q \right)\!,\nn
\end{align}
where $\mathcal{A}_{\mathrm{SBS}}$ is given by~\eqref{corollary2}.
\end{theorem}

Clearly, the Laplace transforms in Theorem~\ref{theoremCoverageProbability} follow the same expression, however, the argument changes depending on the tier including the associated BS. The presentation of the Laplace transform of the interference $\mathcal{I}$ follows.

\begin{proposition}\label{LaplaceTransform} 
The Laplace transform of the interference power from the $i$th BS of the MBS tier when the BSs have multiple antennas and serve multiple users, is given by
\begin{align}
&\mathcal{L}_{I}\!\left({s} \right)=\exp\bigg\{\!\!-s^{\frac{2}{\al}}\!\!\sum_{j= \{s,m\}}\!\lambda_{j} {p_{j}}^{\frac{2}{a}}\mathcal{C}\left( \al, \Psi_{j},w_{j}\right)\!\!\bigg\},
\end{align}
where $\mathcal{C}\left( \al, \Psi_{j},w_{j}\right)=\displaystyle\frac{2 \pi }{\al}\sum_{i=1}^{\Psi_{j}}\!\!\binom{\Psi_{j}}{i}\mathrm{B}^{'}\!\!\left(\! \Psi_{j}\!+\!i\!-\!\frac{2}{a},i\!+\!\frac{2}{a}, w_{j}\right)$ with $B_{x}^{'}\left( p,q \right)$ being the complimentary incomplete Beta function defined in Appendix~\ref{LaplaceTransformproof}. Note that $w_{j}$ represents the minimum distance between the typical user and its nearest interfering BS in tier $j$. In the case that the user is associated with the MBS, $w_{m}=\frac{1}{1+\left( s p_{m} \right)^{-\al}r}$ and $w_{s}=\frac{1}{1+\left( s p_{s} \right)^{-\al}\eta^{-1}}$ with $\eta=1/\left( \hat{p}_{s}\hat{\Delta}_{s} \hat{B}_{s}r^{{\al}} \right)$. Otherwise, if the user is associated with the SBS, $w_{m}=\frac{1}{1+\left( s p_{m} \right)^{-\al}\eta}$ with $\eta=1/\left( \hat{p}_{s}\hat{\Delta}_{s} \hat{B}_{s}r^{{\al}} \right)$ and $w_{s}=\frac{1}{1+\left( s p_{s} \right)^{-\al}r}$.
\end{proposition}
\begin{proof}
See Appendix~\ref{LaplaceTransformproof}.
\end{proof}

The $k$th derivative of the Laplace transform is provided by means of the following lemma.

\begin{lemma}\label{DerivativeOfLaplace} 
The $k$th derivative of the Laplace transform of the interference plus noise power is given by 
\begin{align}
\frac{\mathrm{d}^{k}\!\mathcal{L}_{\mathcal{I}N}\{s\}}{\mathrm{d}s^{k}}\!=\!\mathcal{L}_{\mathcal{I}N}\{s\}\!\!\sum_{m=1}^{k}\sum_{\substack{i_{1}+i_{2}+\ldots+i_{k}=m\\
i_{1}+2i_{2}+\ldots+ik_{i}=k}
}\!\!\!\!\!\!\!\!\mathcal{C}\left(i_{j} \right)\prod_{j=1}^{k}E\left( j \right)\
\end{align}
where 
\begin{align}
  &\mathcal{C}\left( i_{j} \right)=\frac{k!}{\prod_{m=1}^{k}m!^{i_{j}}m^{i}},\nn\\
 &D_{i}\left( j \right)=\frac{\lambda_i}{\al}\left( -1 \right)^{j}\frac{\left( \Psi_{i}+j-1 \right)~} {\left(\Psi_{i}-1 \right)!}\mathrm{B}^{'}\!\left(\! \Psi_{j}\!+\!\frac{2}{a},l\!-\!\frac{2}{a}, w_{j}\right),\nn\\
&E_{j}\left( j \right)=\left( -N\mathbbm{1}_{j=1}+2\pi\sum_{i=1}^{2}\left( -1 \right)^{j}D_{i}\left( j \right)p_{i}^{j}\left( s p_{i} \right)^{\frac{2}{\al}-j}\right)^{m_{j}}\nn.
\end{align}
\end{lemma}
\begin{proof}
See Appendix~\ref{DerivativeOfLaplaceproof}.
\end{proof}

\subsubsection{Cooperative scenario}
\begin{theorem}\label{theoremCoverageProbabilityCooperative} 
In the cooperative model, the downlink probability of coverage $P_{c,D_{m}}\left( \mathcal{T},q \right)$ of the typical user associated with the $i$th BS in the MBS tier, is obtained by substituting the PDF of the distance $f_{R_{cm}}\left( r \right)$ instead of $f_{R_{m}}\left( r \right)$. In the case that the typical user is associated with the SBS cluster, its downlink probability of coverage $P_{c,D_{s}}\left( \mathcal{T},q \right)$ is given by 
\begin{align}
 P_{c,D_{s}}\left( \mathcal{T},q \right)= &\int \displaylimits_{\substack{0<r_{1,s}<r_{2,s}\ldots\\r_{K,s}<r_{ \infty,s}}}\sum^{\Delta_{m}-1}_{k=0}\frac{1}{k!}\left( -s \right)^{k}\nn\\
 &\times \frac{\mathrm{d}^{k}}{\mathrm{ds}^{k}}\left( \mathrm{e}^{-sN}\mathcal{L}_{\mathcal{I}}\{s \} \right) f_{R_{cs}}\left( \br \right)\mathrm{d}\br,
\end{align}
where $s=\frac{\mathcal{T}}{\sum_{x_{i,s}\in\mathcal{C}}^{x_{K,s}}\hat{p}_{s}\hat{\Delta}_{s}\hat{B}_{s}r_{i,s}^{-\al}}$ and the Laplace transform is obtained by Proposition~\ref{LaplaceTransform} with specific limits. Specifically, when the user is associated with the MBS, we have $w_{m}=\frac{1}{1+\left( s p_{m} \right)^{-\al}r}$ and $w_{s}\approx\frac{1}{1+\left( s p_{s} \right)^{-\al}\eta^{-1}}$ with $\eta=1/\left( \hat{p}_{s}\hat{\Delta}_{s} \hat{B}_{s}r^{{\al}} \right)$. In the case that the user is associated with the SBS cluster, we have $ \sum_{x_{i,s}\in\mathcal{C}}^{x_{K,s}}\hat{p}_{s}\hat{\Delta}_{s}\hat{B}_{s}r_{i,s}^{-\al} >r_{m}^{\al}$. The first interference is found at a distance $r_{1,s}^{-\al} >\beta r_{m}^{\al}$. Thus, we have $w_{m}=\frac{1}{1+\left( s p_{m} \right)^{-\al}\eta}$ with $\eta=1/\left( \hat{p}_{s}\hat{\Delta}_{s} \hat{B}_{s}r_{1,s}^{\al} \right)$ and $w_{s}=\frac{1}{1+\left( s p_{s} \right)^{-\al}r_{1,s}}$.

The overall coverage probability of the cooperative scenario is given by
\begin{align}
P_{c,cl}\!\left( \mathcal{T},q \right)\! =\!\left( 1\!-\!\mathcal{A}_{\mathrm{SBS_{cl}}}\right)P_{c,D_{m}}\!\left( \mathcal{T},q \right)\!+\!\mathcal{A}_{\mathrm{SBS_{cl}}} P_{c,D_{s}}\!\left( \mathcal{T},q \right)\!,\nn
\end{align}
where $\mathcal{A}_{\mathrm{SBS_{cl}}}$ is given by Proposition~\ref{PropAssoc}\footnote{Notably, if SDN and WNV are not introduced, the practical implementation of this complicated system is questionable.}.
\end{theorem}
\begin{proof}
The proof follows similar lines to Theorem~1, hence skipped.
\end{proof}

\subsection{Mean Achievable Rate}\label{AverageAchievableRate1}
Herein, having obtained the conditional coverage probability, we derive the mean achievable rate for both noncooperative and cooperative models.
\begin{theorem}\label{AverageAchievableRate1Theorem}
 The mean achievable rate in a multi-antenna HetNet capable of cell association and offloading, corresponding to the noncooperative scenario, is given by 
 \begin{align}
\!\!\!R_{nc}\!\left( \mathcal{T},q \right)\! =\!\left( 1\!-\!\mathcal{A}_{\mathrm{SBS}} \right)R_{A_{m}}\!\!\left( \mathcal{T},q \right)\!+\!\mathcal{A}_{\mathrm{SBS}} R_{A_{s}}\!\!\left( \mathcal{T},q \right)\!,
\end{align}
where $\mathcal{A}_{\mathrm{SBS}}$ is given by~\eqref{corollary2}. The mean achievable rate for the cooperative scenario is written as
\begin{align}
R_{c, cl}\left( \mathcal{T},q \right) =\left( 1-\mathcal{A}_{\mathrm{SBS_{cl}}}\right)R_{D_{m}}\left( \mathcal{T},q \right)+\mathcal{A}_{\mathrm{SBS_{cl}}} R_{D_{s}}\left( \mathcal{T},q \right),
\end{align}
where $\mathcal{A}_{\mathrm{SBS_{cl}}}$ is provided by Proposition~\ref{PropAssoc}, while $R_{i}$ for $i=A_{m},~A_{s},~D_{m},~D_{s}$ is obtained by substituting $P_{c,i}$ into 
\begin{align}
 R_{i}=\frac{1}{\ln 2}\int_{0}^\infty \!\!\!P_{c,i}\,\frac{1}{1+\theta}\mathrm{d}\theta,
\end{align}
where $P_{c,i}$ is obtained by means of Theorems~\ref{theoremCoverageProbability} and~\ref{theoremCoverageProbabilityCooperative}.
\end{theorem}
\proof See Appendix~\ref{AverageAchievableRate1TheoremProof}.\endproof

  {\section{Numerical Results} \label{Numerical} 
In this section, we investigate and compare the traffic offloading performance in the case of cooperative and noncooperative models where the BSs have multiple antennas. Specifically, we shed light into the variation of the coverage probability and  rate provided by Theorems~\ref{theoremCoverageProbability},~\ref{theoremCoverageProbabilityCooperative}, and~\ref{AverageAchievableRate1Theorem}, respectively, with respect to the system parameters. Remarkably, we illustrate the outperformance of the cooperative cluster against the ``selfish'' noncooperative strategy\footnote{Note that the penalty will be an extra cost and complexity at both the transmit and the receive sides.}. Also, we confirm that it is better to serve a single user instead of serving multiple users e.g., by means of SISO or SU-BF. Actually, the employment of more antennas for multi-stream transmission is not always beneficial as the lines corresponding to SDMA transmission show. This work brings to light many new observations,  {while we consider practical values for the simulation parameters used in the literature~\cite{Dhillon2013,Jo2012,Gupta2014}}.  For example, with comparison to~\cite{Dhillon2013} we illustrate how the coverage probability varies with the SINR when the SBS cooperate as well as we have cell association.}

 {Having assumed that the typical user is found at the origin, we choose a sufficiently large area of $5~\mathrm{km}\times 5$ $\mathrm{km}$ including the (two) tiers of MBSs and SBSs. All the BSs have a wired connection with the SDN controller. The locations of the BSs in each tier are simulated as realizations of a PPP with given density. In fact, we assume $\lambda_{m}=0.01~\mathrm{m^{-2}}$ and $\lambda_{s}=0.04~\mathrm{m^{-2}}$ for the MBS and SBS tiers, respectively. Moreover, the users' PPP density in each tier is considered to be $\lambda_{k}=40 \lambda_s$. Regarding the number of cooperating SBSs consisting the cluster, we assume that it equals 2 because a higher number might be impractical. The path-loss for both tiers is set to $\al=3$. Also, the downlink transmit powers are $p_{m}=45~\mathrm{dBm}$ and $p_{s}=35~\mathrm{dBm}$ since a MBS normally operates in higher power than SBSs. In the case of cooperation, we assume $K=2$ SBSs. Continuing to the design of the setup, we consider $3$ antenna configurations holding simultaneously for both tiers and defining the corresponding transmission strategies~\footnote{Taking into account mixed transmission strategies, i.e., SDMA transmission for the first tier, and SISO for the second one, has been omitted due to limited space, and because the investigation of such cases is of minor importance.}. In particular, we have
\begin{itemize}
 \item SISO: $M_{m}=\Psi_{m}=1$ and $M_{s}=\Psi_{s}=1$
 \item SUBF: $M_{m}=8$, $\Psi_{m}=1$, $M_{s}=4$, $\Psi_{s}=1$
 \item SDMA: $M_{m}=\Psi_{m}=8$ and $M_{s}=\Psi_{s}=8$.
\end{itemize}}

 {The simulation of the metrics under study, i.e., the coverage probability and the mean achievable rate of the typical user necessitate the calculation of the received SINR in terms of the desired signal strength and interference power from each BS. For example, coverage of the user means that the received SINR from at least one of the BSs of both tiers exceeds a certain target. This procedure, being repeated $10^{4}$ times, provides the validation of our model, and demonstrates the performance of the system by varying its parameters. Actually, in the figures, we have depicted the proposed analytical expressions of the coverage probability $P_{c}\left( \mathcal{T},q \right)$ and the achievable user rate $R_{c}\left( \mathcal{T},q \right)$ by means of lines being ``solid'' and ``dotted'' when the cooperative and noncooperative architectures are deployed, respectively. In addition, we have included the corresponding simulated results designated by means of bullets. }

\begin{figure}[!h]
 \begin{center}
 \includegraphics[width=0.8\linewidth]{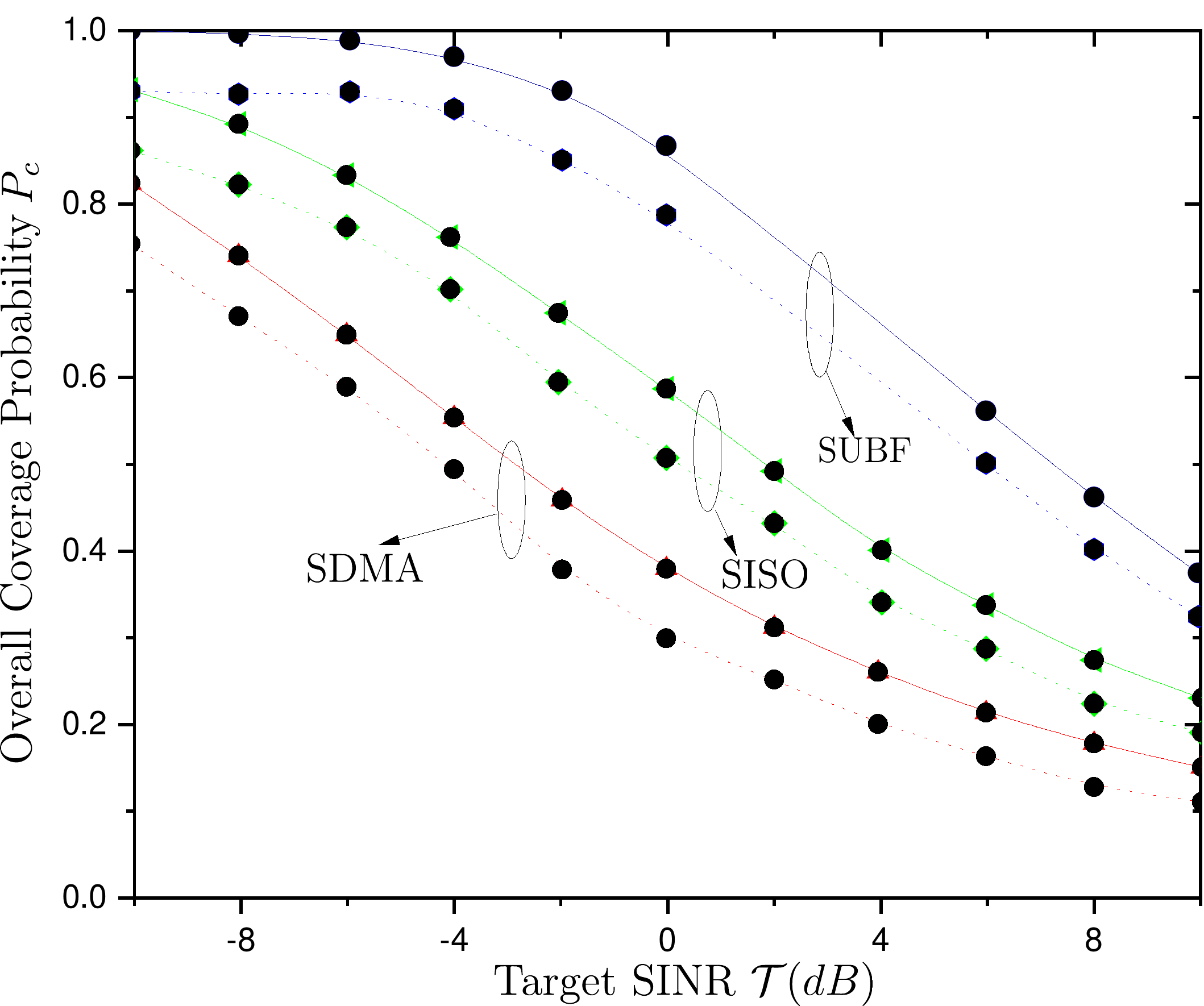}
 \caption{\footnotesize{Overall coverage probability of a MU-MIMO HetNet for varying transmission strategies versus the target SINR $\mathcal{T}$ for both noncooperative and cooperative scenarios,  {where $\lambda_{m}=0.01~\mathrm{m^{-2}}$,  $\lambda_{s}=0.04~\mathrm{m^{-2}}$, $\lambda_{k}=40 \lambda_s$, $p_{m}=45~\mathrm{dBm}$, $p_{s}=35~\mathrm{dBm}$, $M_{m}=\Psi_{m}=1$ and $M_{s}=\Psi_{s}=1$ (SISO), $M_{m}=8$, $\Psi_{m}=1$, $M_{s}=4$, $\Psi_{s}=1$ (SUBF), $M_{m}=\Psi_{m}=8$ and $M_{s}=\Psi_{s}=8$ (SDMA)}.}}
 \label{Fig1}
 \end{center}
 \end{figure}
 \begin{figure}[!h]
 \begin{center}
 \includegraphics[width=0.8\linewidth]{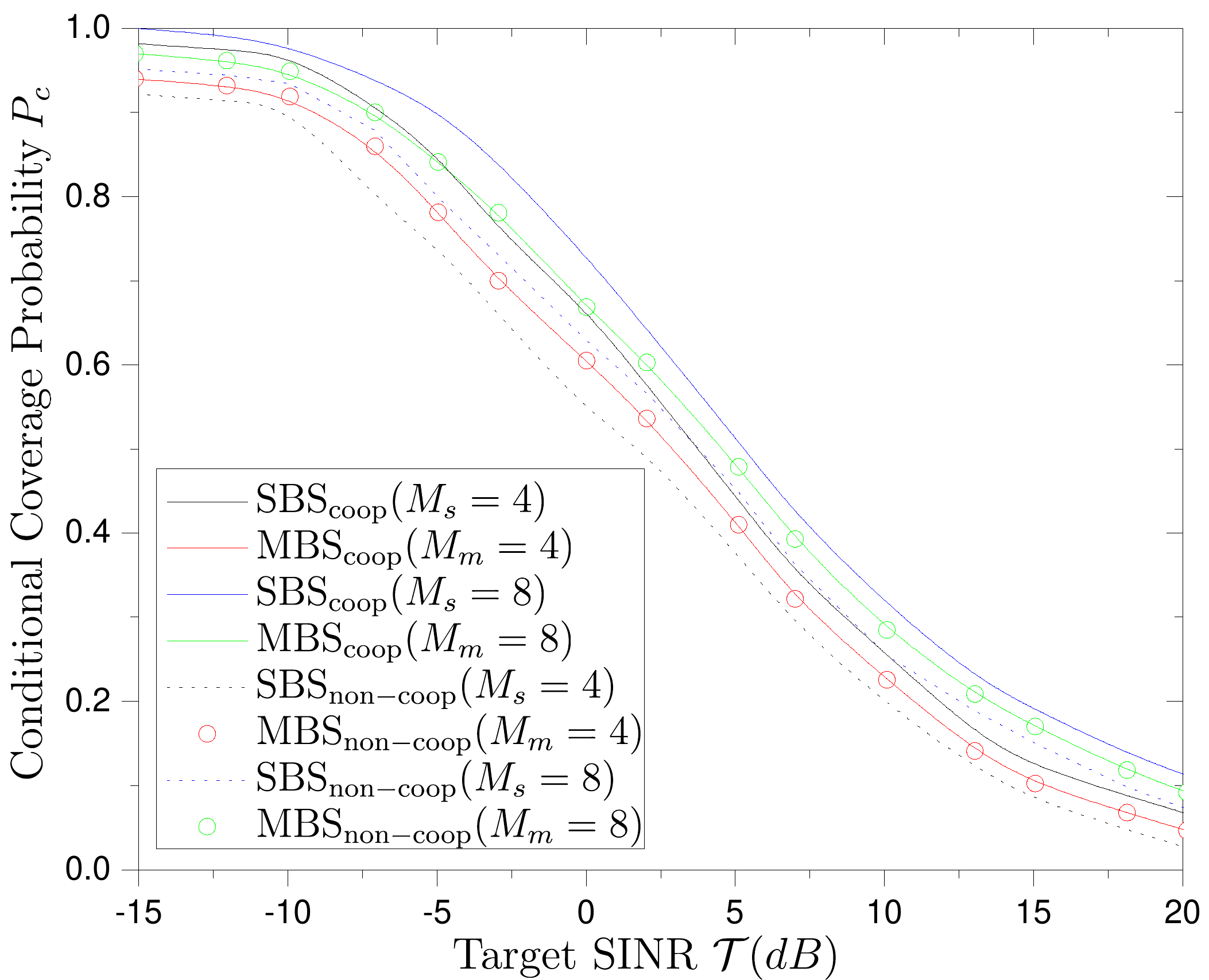}
 \caption{\footnotesize{Coverage probabilities of a MU-MIMO HetNet for varying number of BSs antennas versus the target SINR $\mathcal{T}$ for both noncooperative and cooperative scenarios,  {where $\lambda_{m}=0.01~\mathrm{m^{-2}}$,  $\lambda_{s}=0.04~\mathrm{m^{-2}}$, $\lambda_{k}=40 \lambda_s$, $p_{m}=45~\mathrm{dBm}$, $p_{s}=35~\mathrm{dBm}$.}}}
 \label{Fig2}
 \end{center}
 \end{figure}
 {Fig.~\ref{Fig1} depicts the overall probability of coverage with respect to the target SINR for different transmission strategies under the options of cooperation and noncooperation among the SBSs. Obviously, SDMA is not preferable with a comparison to SUBF because of the interfering users. Also, between SISO and SDMA, the latter is inferior because more interfering BSs exist. Also, SUBF appears the best coverage because of its inherent beamforming gain in addition to the proximity gain obtained by SISO. These observations come into agreement with simpler results known in the literature, e.g.,~\cite{Dhillon2013}. The new finding concerns that this behavior is met, while offloading takes place in both scenarios of SBS operation. Notably, these design attributes are achievable due to the introduction of SDN in the proposed system. Moreover, the SBS cluster behaves advantageously, while if SBSs do not cooperate we observe lower coverage. As a result, the scalability of the SDN controller in wireless networks contributes to increased coverage.}

 {In Fig.~\ref{Fig2}, we compare the conditional coverage probabilities corresponding to the events $A_{m}$ $A_{s}$, $D_{m}$, and $D_{s}$. In particular, when the SBSs do not collaborate, the coverage provided by the MBS and the single SBS would be identical if we had considered that the transmit powers $p_{m}$ and $p_{s}$ are equal, as you can see from~\eqref{theorem1_1} and~\eqref{theorem1_2}. In our scenario, we assume that the transmit powers are different ($p_{m}>p_{s}$). In such case, the typical user selects the BS providing the largest biased power. Herein, we also assume equal bias, hence the user selects the MBS. On the other hand, when the SBSs join forces thanks to the SDN controller, they enhance the coverage since their transmitted power is higher than the power emitted by the MBS. However, the coverage obtained by the MBS is the same with the noncooperative case. Moreover, the coverage is improved, if only the number of BS antennas increases due to higher beamforming gain. As a corollary, the increase of the BS antennas and SBS cooperation are indicated when better coverage is demanded.}

 {The maximization of the overall coverage with the aim to find the optimal selection bias is examined in Fig.~\ref{Fig3}. 
% Having in mind that BS selection based on the highest mean SINR may not always maximize coverage
Especially, the figure demonstrates the overall probability of coverage versus the relative bias $B_{2}/B_{1}$ when the threshold SINR is $0$ dB. We show the maximization of $P_{c}$ for different transmission strategies, i.e., the SDMA provides the worst coverage due to the intra-cell interference, while SUBF behaves best. Interestingly, while focusing on SDMA and SUBF, the optimal bias moves to the left since the received signal from the SBS cluster is higher and allows lower bias for maximum coverage. Hence, higher preference is shown to the connection with the SBS tier. Notably, in the case of SBS cooperation, the optimal bias moves further right since the typical user selects the SBS cluster to communicate. Moreover, the gap between the dotted and solid lines, i.e., between the noncooperative and cooperative scenarios increases as $B_{2}$ increases over $B_{1}$.\footnote{  {Generally, BS cooperation among multiple BSs can improve the average rate for users at the cell edge. In such case, multiple BSs use the same radio resources to serve one cell-edge user. However, in the case of no BS cooperation, each BS could use these radio resources to serve one user. Hence, this topic could be investigated in terms of fairness but due to limited space, we leave its study for future research.}}}

 {Fig.~\ref{Fig4} depicts the impact of the MBS transmit power and SBS density on the overall coverage probability. As the MBS power $p_{m}$ increases, it would be expected that the user would connect with a MBS, however, the SBS cluster still provides better coverage due to the coordination by the SDN controller. Regarding the SBS density, we observe that higher density results in higher coverage probability since the connection between the typical user and the SBS cluster is more favorable. Similar observation holds for the noncooperative scenario because it is more likely that a SBS is closer to the user than a MBS. Especially, in higher SBS density the coverage probability saturates because after a certain large value of $\lambda_{s}$ over $\lambda_{m}$, the user will associate with the SBS or the SBS cluster in the noncooperative and cooperative cases, respectively with no higher possibility. In other words, any further increase in the SBS density will be fruitless.}
\begin{figure}[!h]
 \begin{center}
 \includegraphics[width=\linewidth]{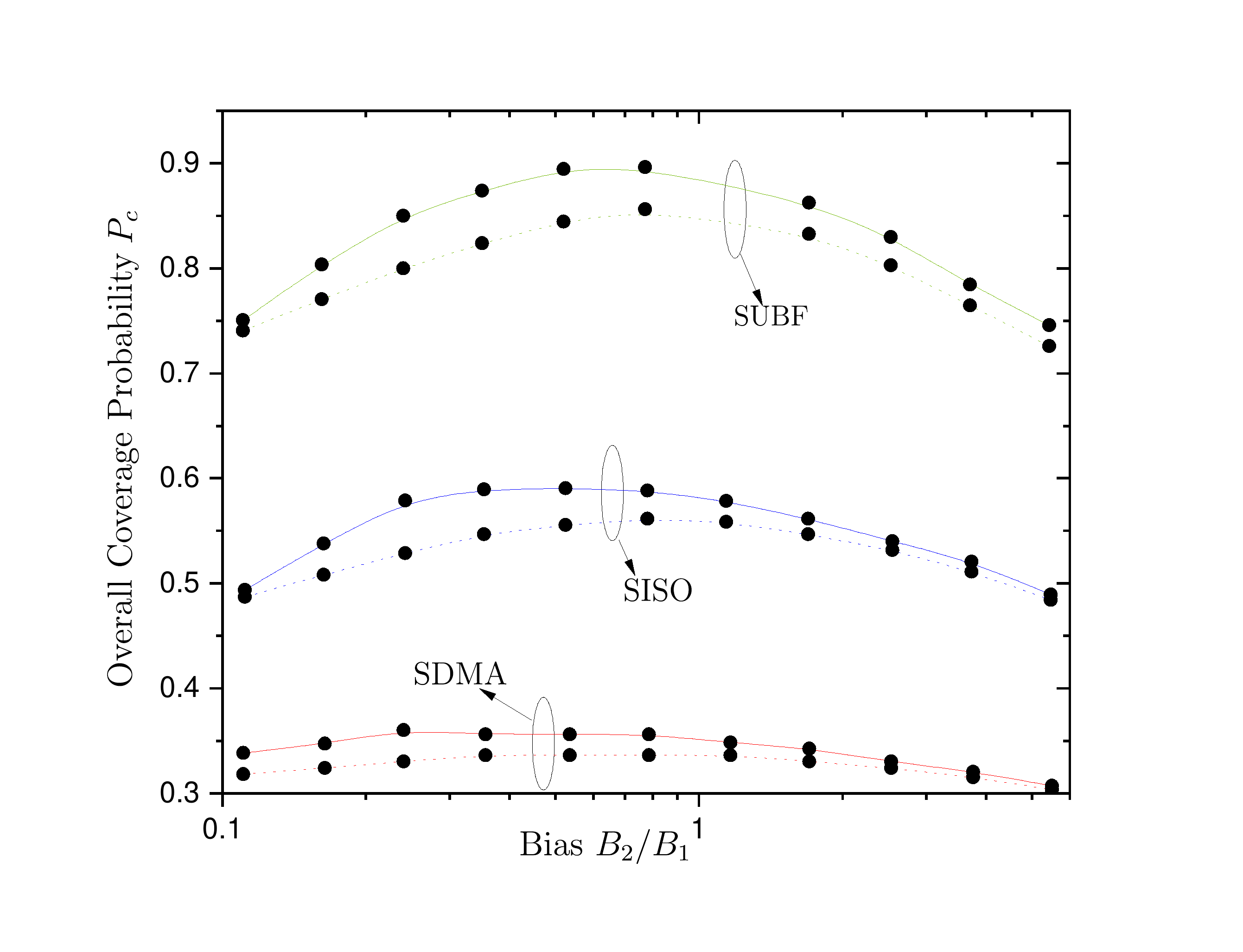}
 \caption{\footnotesize{Overall coverage probability of a MU-MIMO HetNet for varying transmission strategies versus the relative bias $B_{s}/B_{m}$ for both noncooperative and cooperative scenarios,  {where $\lambda_{m}=0.01~\mathrm{m^{-2}}$,  $\lambda_{s}=0.04~\mathrm{m^{-2}}$, $\lambda_{k}=40 \lambda_s$, $p_{m}=45~\mathrm{dBm}$, $p_{s}=35~\mathrm{dBm}$, $M_{m}=\Psi_{m}=1$ and $M_{s}=\Psi_{s}=1$ (SISO), $M_{m}=8$, $\Psi_{m}=1$, $M_{s}=4$, $\Psi_{s}=1$ (SUBF), $M_{m}=\Psi_{m}=8$ and $M_{s}=\Psi_{s}=8$ (SDMA)}.}}
 \label{Fig3}
 \end{center}
 \end{figure}
 \begin{figure}[!h]
 \begin{center}
 \includegraphics[width=0.8\linewidth]{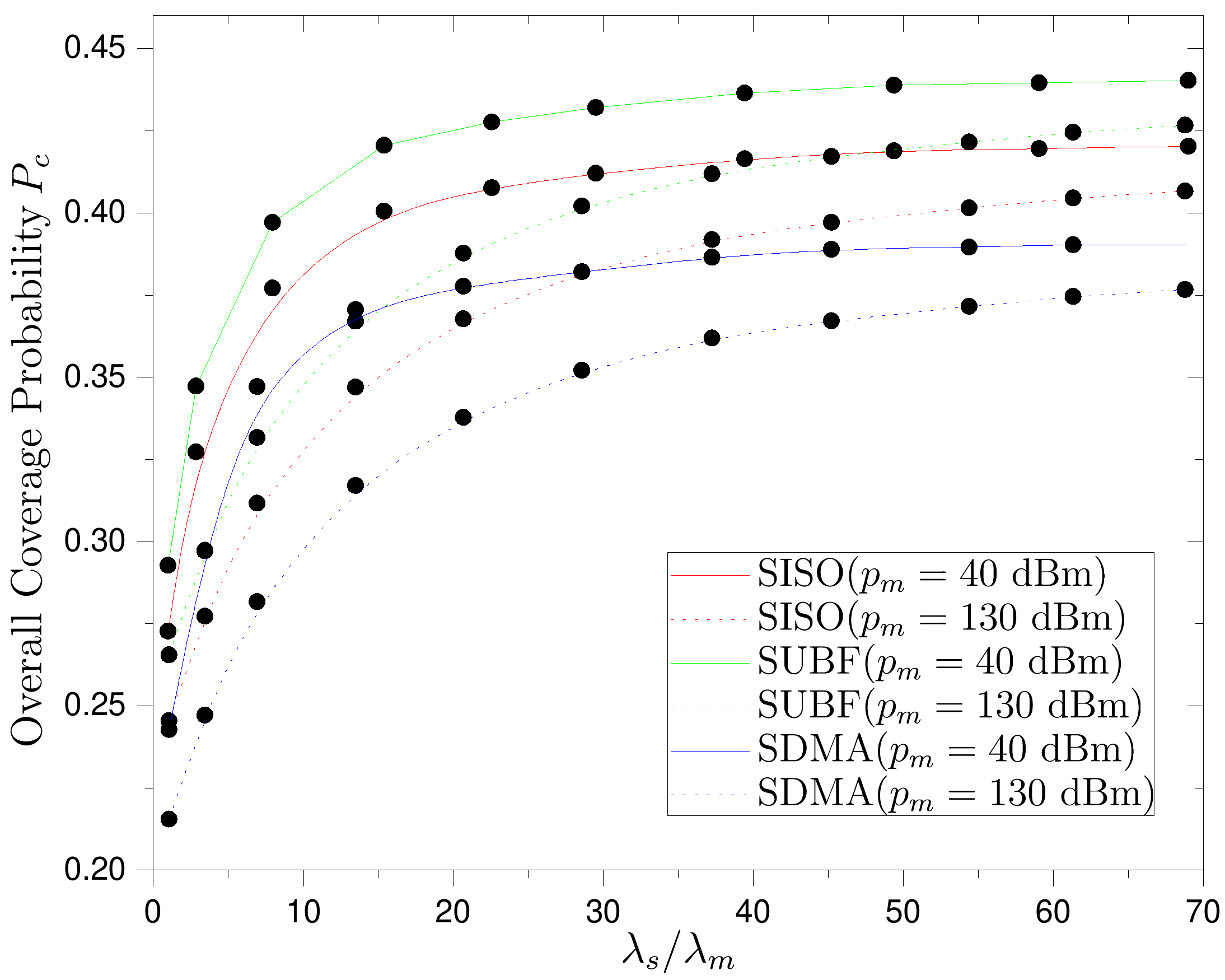}
 \caption{\footnotesize{Overall coverage probability of a MU-MIMO HetNet for varying transmission strategies versus the relative density $\lambda_{s}/ \lambda_{m}$ for both noncooperative and cooperative scenarios,  {where $\lambda_{k}=40 \lambda_s$,  $p_{s}=35~\mathrm{dBm}$, $M_{m}=\Psi_{m}=1$ and $M_{s}=\Psi_{s}=1$ (SISO), $M_{m}=8$, $\Psi_{m}=1$, $M_{s}=4$, $\Psi_{s}=1$ (SUBF), $M_{m}=\Psi_{m}=8$ and $M_{s}=\Psi_{s}=8$ (SDMA)}.}}
 \label{Fig4}
 \end{center}
 \end{figure}
 \begin{figure}[!h]
 \begin{center}
 \includegraphics[width=0.8\linewidth]{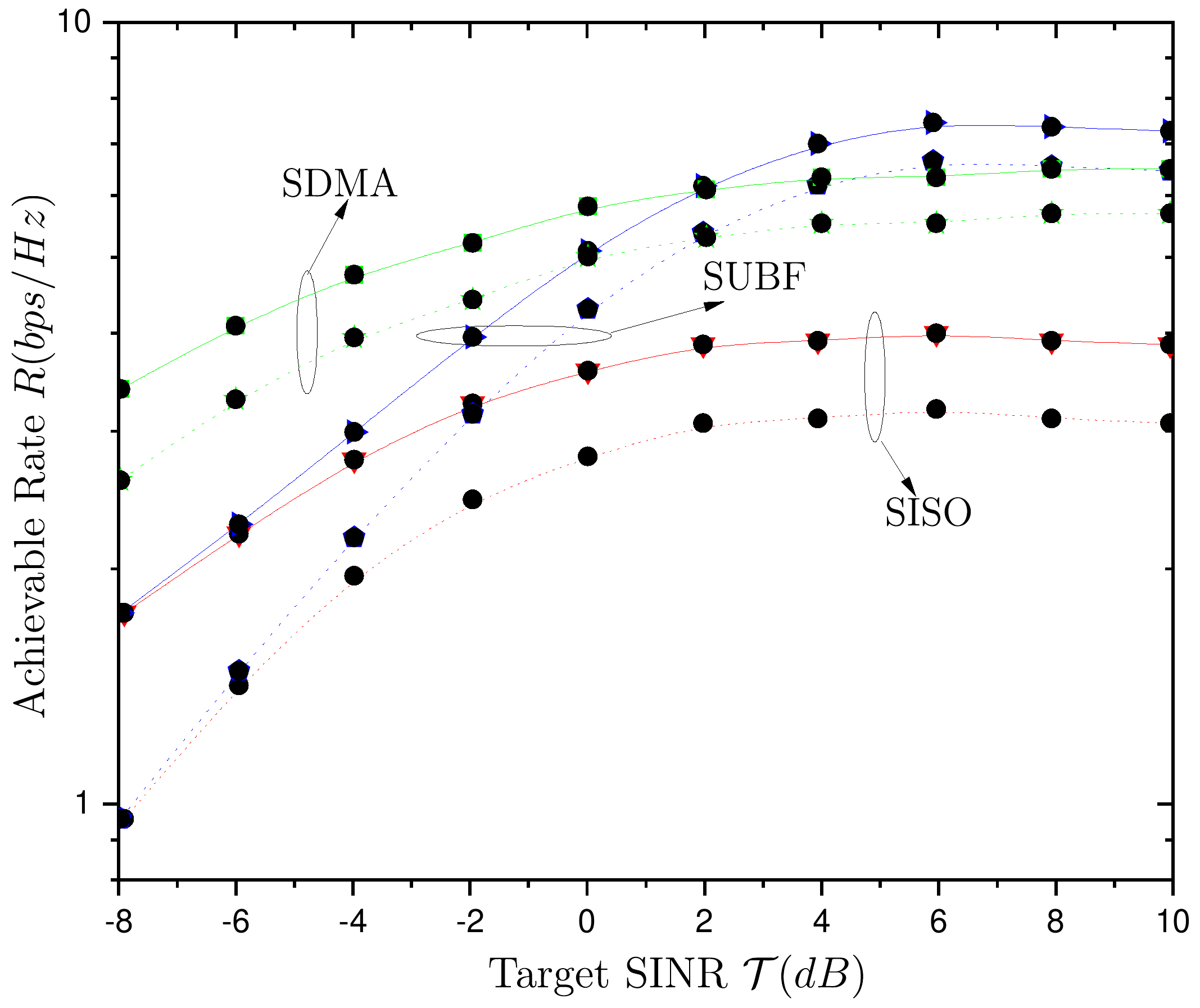}
 \caption{\footnotesize{Overall mean achievable rate of a MU-MIMO HetNet for varying transmission strategies versus the target SINR $\mathcal{T}$ for both noncooperative and cooperative scenarios,  {where $\lambda_{m}=0.01~\mathrm{m^{-2}}$,  $\lambda_{s}=0.04~\mathrm{m^{-2}}$, $\lambda_{k}=40 \lambda_s$, $p_{m}=45~\mathrm{dBm}$, $p_{s}=35~\mathrm{dBm}$, $M_{m}=\Psi_{m}=1$ and $M_{s}=\Psi_{s}=1$ (SISO), $M_{m}=8$, $\Psi_{m}=1$, $M_{s}=4$, $\Psi_{s}=1$ (SUBF), $M_{m}=\Psi_{m}=8$ and $M_{s}=\Psi_{s}=8$ (SDMA)}.}}
 \label{Fig5}
 \end{center}
 \end{figure}
 
 {Fig.~\ref{Fig5} shows a comparison between the cooperative and noncooperative designs in terms of the achievable rate while the SINR varies. In addition, different transmission strategies are explored under the different designs. The SISO transmission is not preferable concerning the achievable rate, and the SDMA technique saturates at high SINR due to interference, as expected. Furthermore, the SBS cooperation, managed by means of SDN, offers the advantage of a higher rate than the ``selfish'' strategy for all transmission techniques.}

 \section{Conclusion} \label{Conclusion} 
 In this paper, we investigated the downlink coverage probability and mean achievable rate of MIMO HetNets with flexible cell association, and mainly, the ability of SBSs to cooperate by means of SDN. Specifically, an SDN controller was introduced to the architecture to alleviate the burden of the system by undertaking tasks such as the cell association and SBS coordination. Embodying the benefits of BS cooperation, multiple-antenna transmission, and offloading, we derived the overall coverage probability and mean achievable rate. Numerical results, verified by Monte Carlo simulations, revealed that the offloading of users is enhanced when the SBSs cooperate. In addition, among the results, we showed that SBS cooperation increases the coverage probability and rate under different transmission techniques. Also, the system performance is improved by increasing the number of BS antennas due to higher beamforming gain, while SBSs cooperate. Finally, the SBS density is meaningful to be increased to a certain value, because extra increase has no benefit.
\begin{appendices}
\section{Proof of Proposition~\ref{PropRegion}}\label{Alpha}
The association region between the typical user and the
cooperative cluster and not with the MBS is obtained, if the following inequality is satisfied. Specifically, we have
\begin{align}
 \sum_{x_{i,s}\in \mathcal{C}}^{x_{K,s}}\bar{p}_{\mathrm{r}_{i,s}}B_{s}&>\bar{p}_{\mathrm{r}_{m}}B_{m} \\
 \sum_{x_{i,s}\in\mathcal{C}}^{x_{K,s}}p_{s}\Delta_{s}r_{i,s}^{-\al}B_{s}&> p_{m}\Delta_{m} r_{m}^{-\al_{m}}B_{m}\label{label1}\\
 r_{m}&>\left( \sum_{x_{i,s}\in\mathcal{C}}^{x_{K,s}}\hat{p}_{s}\hat{\Delta}_{s}\hat{B}_{s}r_{i,s}^{-\al_{s}} \right)^{\frac{1}{{\al}_{m}}},
\end{align}
where $\hat{p}_{s}={p}_{s}/p_{m}$, and similarly  $\hat{\Delta}_{s}={\Delta}_{s}/\Delta_{m}$, $\hat{B}_{s}={B}_{s}/B_{m}$. Note that in~\eqref{label1}, we have substituted~\eqref{meanPower}. Moreover, $r_{i,s}$ and $r_{m}$ express the distances from the typical user to its nearest $i$th closest SBS of the cluster  and MBS, respectively. 
\section{Proof of Proposition~\ref{PropAssoc}}\label{AlphaAssoc}
The user will associate with the tier offering the highest biased received power. In fact, the probability that the user will communicate with the MBS is $\displaystyle\mathbb{P}\left(\bar{p}_{\mathrm{r}_{m}}B_{m}>\sum_{x_{i,s}\in \mathcal{C}}^{x_{K,s}}\bar{p}_{\mathrm{r}_{i,s}}B_{s} \right)$, while its complement is the association probability with the SBS cluster. Specifically, we have
\begin{align}
&\mathcal{A}_{\mathrm{SBS_{cl}}}=1- \mathbb{P}\left(\bar{p}_{\mathrm{r}_{m}}B_{m}>\sum_{x_{i,s}\in \mathcal{C}}^{x_{K,s}}\bar{p}_{\mathrm{r}_{i,s}}B_{s} \right)\\&=1-\EE_{\br}\left[\mathbb{P}\left(r_{m}<\left(1\bigg / \sum_{x_{i,s}\in \mathcal{C}}^{x_{K,s}}\hat{p}_{s}\hat{\Delta}_{s}\hat{B}_{s}r_{i,s}^{-\al} \right)^{\frac{1}{\al}} \right) \right] \label{PropAssoc1}\\
&=\int \displaylimits_{\substack{0<r_{1,s}<r_{2,s}\ldots\\r_{K,s}<r_{ \infty,s}}}\!\!\!\!\!\!\!\!\mathbb{P}\!\left(\!r>\left(\frac{1}{ \hat{{p}}_{s}{\hat{\Delta}}_{s}{\hat{B}}_{s} \sum_{x_{i,s}\in \mathcal{C}}^{x_{K,s}}r_{i,s}^{-\al} } \right)^{\frac{1}{\al}} \right)\! f_{\Lambda}\!\left( \br \right)\mathrm{d}\br\label{PropAssoc2}\\
&=\int \displaylimits_{\substack{0<r_{1,s}<r_{2,s}\ldots\\r_{K,s}<r_{ \infty,s}}}\mathrm{e}^{\left( -\lambda_{m}\pi\eta^{2/\al} \right)}f_{\Lambda}\left( \br \right)\mathrm{d}\br,\label{PropAssoc3}
\end{align}
where in~\eqref{PropAssoc1} $\br=[r_{1,s},\ldots,r_{K-1,1},r_{K,1}]$ denotes the set of distances of the $K$ closest SBSs to the typical user, while in~\eqref{PropAssoc2} $f_{\Lambda}\left( \br \right)$ is the joint PDF of $r_{i,s}$. In the same equation, we have assumed that all SMBs have the same parameters. Next, $r_{m}>\eta ^{\frac{1}{\al}} $ with $\eta=1/\left( \hat{p}_{s}\hat{\Delta}_{s} \hat{B}_{s}\sum_{x_{i,s}\in\mathcal{C}}^{x_{K,s}}r_{i,s}^{-{\al}} \right)$ expresses that there is no SBS in the circle with radious $\eta ^{\frac{1}{\al}} $, and the corresponding probability is actually the null probability of a $2$-D homogeneous PPP. Given that the null probability of a $2$-D homogeneous PPP with density $\lambda$ in an area $A$ is $\mathrm{exp}^{-\lambda A}$, we obtain~\eqref{PropAssoc3}, which concludes the proof. 

\section{Proof of Corollary~\ref{corollary1}}\label{AlphaAssocCorollary1}
The proof starts similarly to the proof of~Proposition~\ref{PropAssoc}. Hence, we note that $\mathcal{A}_{\mathrm{SBS}}$ is obtained when the biased received power from the single SBS is greater than the power received by the MBS. Mathematically, we have
\begin{align}
\mathcal{A}_{\mathrm{SBS}}&= \mathbb{P}\left(\bar{p}_{\mathrm{r}_{s}}B_{s}>\bar{p}_{\mathrm{r}_{m}}B_{m} \right)\\
&=\EE_{r_{m}}\left[\mathbb{P}\left(r_{s}<\left(\hat{p}_{s}\hat{\Delta}_{s}\hat{B}_{s}r_{m}^{\al_{m}} \right)^{\frac{1}{\al}} \right)\right]\label{appC34} \\
&=\int_{0}^{\infty} \mathbb{P}\left(r_{s}<\left(\hat{p}_{s}\hat{\Delta}_{s}\hat{B}_{s}r^{\al_{m}} \right)^{\frac{1}{\al}}\right)f_{r_{m}}\left( r \right)\mathrm{d}r\label{AlphaAssocCor1}.
\end{align}
Taking into account the expression for the null probability of a $2$-D homogeneous PPP with density $\lambda_{s}$, we have
\begin{align}
 \mathbb{P}\left(r_{s}<\left(\hat{p}_{s}\hat{\Delta}_{s}\hat{B}_{s}r^{\al_{}} \right)^{\frac{1}{\al_{}}} \right)=1-\mathrm{e}^{-\lambda_{s}\pi r^{2} \left( \hat{p}_{s}\hat{\Delta}_{s}\hat{B}_{s}\right)^{\frac{2}{\al}}}.\label{AlphaAssocCor2} 
\end{align}
Furthermore, according to~\cite{Baumstark2007}, we have
\begin{align}
f_{r_{m}}\left( r \right)&=1-\frac{\mathrm{d}\mathbb{P}\left[ r_{m}>r\right] }{\mathrm{d}r}\nn\\
&=2\pi\lambda_{m}r\mathrm{e^{-\pi\lambda_{m}r^{2}}}.\label{AlphaAssocCor3} 
\end{align}
After substituting~\eqref{AlphaAssocCor2} and \eqref{AlphaAssocCor3} into~\eqref{AlphaAssocCor1}, we obtain the desired result by means of~\cite[Eq.~(3.321.4)]{Gradshteyn2007}.
\section{Proof of Proposition~\ref{PDFdistance}}\label{PDFdistanceProof}
The cumulative distribution function, or else the probability of the event of $R_{s}>r$ conditioned on $A_{s}$, is
\begin{align}
 \mathbb{P}\left[ R_{s}>r \right]&=\mathbb{P}\left[ r_{s}>r |A_{s}\right]\label{PDFdistanceProof0} \\
 &=\frac{\mathbb{P}\left[ r_{s}>r ,A_{s}\right] }{\mathbb{P}\left[ A_{s}\right] }\label{PDFdistanceProof1} \\
 &=\frac{\mathbb{P}\left[ r_{s}>r,\bar{p}_{\mathrm{r}_{s}}B_{s}>\bar{p}_{\mathrm{r}_{m}}D_{m}\right] }{\ \mathbb{P}\left[A_{s}\right] }\\
 &=\frac{\int_{r}^{\infty}\mathbb{P}\left[r_{m}>r \beta^{\frac{1}{\al_{}}}\right]f_{r_{s}}\left( r \right)\mathrm{d}r }{\mathcal{A}_{\mathrm{SBS}}},
 \end{align}
where $\beta=\left( \hat{p}_{s}\hat{\Delta}_{s}\hat{B}_{s}  \right)^{-1}$, and $\mathcal{A}_{\mathrm{SBS}}$ is given by~Corollary~\ref{corollary1}. Note that in~\eqref{PDFdistanceProof0} $r_{s}$ denotes the distance of the nearest SBS,  {while $A_{s}$ expresses the event that the user connects to a SBS}.  { In~\eqref{PDFdistanceProof1}, we have considered the definition of the conditional probabilty. }Moreover, $\mathbb{P}\left[r_{m}>r \beta^{\frac{1}{\al_{}}}\right]$ and $f_{r_{s}}\left( r \right)$ are obtained similarly to~\eqref{AlphaAssocCor2}
and~\eqref{AlphaAssocCor3}, respectively. Thus,~\eqref{PDFdistanceSBS} is obtained after differentiation with respect to $r$. The proof of~\eqref{PDFdistanceMBS1} follows the same steps, but the probability is derived conditioned on the event $A_{m}$.\\
In the cooperative scenario, when the user connects with the MBS, the corresponding probability $\mathbb{P}\left[ r_{m}>r |D_{m}\right]$ is obtained as
\begin{align}
 \mathbb{P}\left[ r_{m}>r |D_{m}\right]=\frac{\mathbb{P}\left[ r_{m}>r,D_{m}\right] }{\ \mathbb{P}\left[D_{m}\right] },
\end{align}
where $\mathbb{P}\left[ r_{m}>r,D_{m}\right] $ is given by
\begin{align}
\!\!\mathbb{P}\left[ r_{m}>r,D_{m}\right] &=\mathbb{P}\!\left[ r_{m}>r,\bar{p}_{\mathrm{r}_{m}}B_{m}>\!\sum_{x_{i,s}\in \mathcal{C}}^{x_{K,s}}\bar{p}_{\mathrm{r}_{i,s}}B_{s}\right]\\
&=\int_{r}^{\infty}\mathbb{P}\left[ r< \eta^{\frac{1}{\al}}\right]f_{r_{m}}\left( r \right)\mathrm{d}r,
\end{align}
with $\mathbb{P}\left[ r< \eta^{\frac{1}{\al}}\right]$ being the complement of the null probability of the $2$-D homogeneous PPP given by $\mathrm{e}^{\left( -\lambda_{m}\pi\eta^{2/\al} \right)}$. Moreover, $f_{r_{m}}\left( r \right)$ has already been derived in \eqref{AlphaAssocCor3}. The last step concerns the differentiation of $\mathbb{P}\left[ r_{m}>r,D_{m}\right] $.\\
In particular, in the interesting scenario, where the typical user is associated with the SBS cluster, we have that $\mathbb{P}\left[ r_{s}>r,D_{s}\right]$ is obtained as the integrated function in~\eqref{PropAssoc2}. By differentiating $\mathbb{P}\left[ r_{s}>r|D_{s}\right]$, the proof is concluded.

\section{Proof of Theorem~\ref{theoremCoverageProbability}}\label{theoremCoverageProbabilityProof}
In the case that the typical user is connected with the $i$th BS of the MBS tier, the SINR, after modifying appropriately \eqref{eq:SINR}, is written as
\begin{align}
\gamma_{m}=\frac{p_{m} h_{i,m} r_{i,m}^{-\al}}{\mathcal{I}+N},\label{eq:SINRproof1} 
\end{align}
where we have that $\mathcal{I}=\displaystyle\sum_{j=\{m,s\}}\sum_{i\in\Phi_{j}\backslash r_{i,j}}p_{j} g_{i,j} r_{i,j}^{-\al}$ describes the aggregate interferences from both tiers, which are independent. Hence, $\mathcal{I}$ can be written as a sum of the two interferences, i.e., $\mathcal{I}=\mathcal{I}_{m}+\mathcal{I}_{s}$, where $\mathcal{I}_{m}=\displaystyle\sum_{i\in\Phi_{m}\backslash r_{i,m}}p_{m} g_{i,m} r_{i,m}^{-\al}$ and $\mathcal{I}_{s}=\displaystyle\sum_{i\in\Phi_{s} r_{i,s}}p_{s} g_{i,s} r_{i,s}^{-\al}$. Note that $N=\sigma^{2}$ is the variance of the AWGN, being assumed to be identical across all tiers, i.e., $\sigma^{2}_{m}=\sigma^{2}_{s}=\sigma^{2}$.
According to the definition, given by~\eqref{defin}, we have
\begin{align}
 &P_{c,A_{m}}\left( \mathcal{T},q \right)=\mathbb{P}\left[ \gamma_{m}>\mathcal{T} \right]\label{coverage_definition0}\\
 &=\EE \!\left[\mathbb{P}\left[ \gamma_{m}>\mathcal{T}|x \right]\right]\label{coverage_definition1}\\
 &=\int_{r>0}\mathbb{P}\!\left[ h_{i,m} >p_{m}^{-1}\mathcal{T}r^{\al}\left( \mathcal{I}+N \right)|r\right] f_{R_{m}}\left( r \right)\mathrm{d}r
\end{align}
The left term of the integrand is written as
\begin{align}
 \mathbb{P}\![ h_{i,m} & >p_{m}^{-1}\mathcal{T}r^{\al}\left( \mathcal{I}+N \right)|r]\nn\\
&=\sum^{\Delta_{m}-1}_{k=0}\frac{1}{k!}\EE\left[\left[ -s\left( \mathcal{I}+N \right)\right]^{k}  \mathrm{e}^{-s\left( \mathcal{I}+N \right)}\right]\label{coverage_definition2}\\
 &=\sum^{\Delta_{m}-1}_{k=0}\frac{1}{k!}\left( -s \right)^{k}\frac{\mathrm{d}^{k}}{\mathrm{ds}^{k}}\mathcal{L}_{\mathcal{I}N}\{s\left( \mathcal{I}+N \right)\}\label{coverage_definition3}\\
 &=\sum^{\Delta_{m}-1}_{k=0}\frac{1}{k!}\left( -s \right)^{k}\frac{\mathrm{d}^{k}}{\mathrm{ds}^{k}}\left( \mathrm{e}^{-sN}\mathcal{L}_{\mathcal{I}}\{s \} \right).\label{coverage_definition4}
\end{align}
In~\eqref{coverage_definition2}, we have set that $s=p_{m}^{-1}\mathcal{T}r^{\al}$. Moreover, given that $h_{i,m}$ is Gamma distributed with shape $\Delta_{m}$ and scale $s \left( \mathcal{I}+N \right)$, we have taken into account that the corresponding Gamma CCDF with shape $s$ and scale $\theta$ is provided by~$\mathbb{P}_{h_{i,m}}\left( z \right)=\displaystyle\sum_{i=0}^{s-1}\frac{1}{i!}\left( \frac{z}{\theta} \right)^{i}\mathrm{e}^{\frac{-z}{\theta}}$. Next, in~\eqref{coverage_definition3}, we have applied the definition of the Laplace Transform $\mathbb{E}_{I }\left[ e^{-s I }\left( s I \right)^{i}\right]=s^{i}\mathcal{L}\{t^{i}g_{I }\left( t \right)\}\left( s \right)$ and the Laplace identity $t^{i}g_{I }\left( t \right)\longleftrightarrow \left( -1 \right)^{i}\frac{\mathrm{d}^{i}}{\mathrm{d}^{i}s}\mathcal{L}_{I }\{g_{I }\left( t \right)\}\left( s \right)$. Finally, the overall
coverage probability is obtain by means of application of the law of total probability while the association probabilities are independent.
\section{Proof of Proposition~\ref{LaplaceTransform}}\label{LaplaceTransformproof}
The interference $\mathcal{I}_{j}$ from the BSs of the $j$th tier, where $j=\{m,~s\}$, results by the BSs located outside the ball of radious $r_{j}$, i.e., $B\left( 0,r_{j} \right)$ with $r_{j}=\left(\hat{p}_{s}\hat{\Delta}_{s}\hat{B}_{s}x ^{-\al} \right)^{\frac{1}{{\al}}}$. The Laplace transform of the overall interference is obtained as
\begin{align}
&\!\!\!\mathcal{L}_{\mathcal{I}}=\EE\left[ \mathrm{e}^{-s\mathcal{I}}\right]\\
&=\EE\left[\exp\bigg\{-s \sum_{j\in \{s,m\}} \sum_{y\in\Phi_{j}\backslash B \left( 0,r_{i,j} \right)}p_{j} g_{y,j} \|y\|^{-\al}\bigg\}\right]\nn \\
&=\prod_{j\in \{s,m\}}\EE \left[ \prod_{y\in\Phi_{j}\backslash B \left( 0,r_{i,j} \right)}\exp\bigg\{-s p_{j} g_{y,j} \|y\|^{-\al}\bigg\}\right]\label{laplace1} \\
&=\!\!\!\!\prod_{j\in \{s,m\}}\!\!\!\!\EE \!\!\left[ \prod_{y\in\Phi_{j}\backslash B \left( 0,r_{i,j} \right)}\!\!\!\!\!\!\!\EE_{g_{y,j}}\!\left[ \exp\bigg\{\!-s p_{j} g_{y,j} \|y\|^{-\al}\bigg\}\right]\! \right]\label{laplace2} \\
&=\prod_{j\in \{s,m\}}\EE \left[ \prod_{y\in\Phi_{j}\backslash B \left( 0,r_{i,j} \right)}\frac{1}{\left( 1+s p_{j} \|y\|^{-\al} \right)^{\Psi_{j}}} \right]\label{laplace3}\\
&=\prod_{j\in \{s,m\}}\!\!\!\!\!\exp\bigg\{\!-\lambda_{j}\!\int_{\mathbb{R}^{2}}\!\!\left( \!1-\frac{1}{\left( 1+s p_{j} \|y\|^{-\al} \!\right)^{\Psi_{j}}} \right)\mathrm{d}r\bigg\}\label{laplace4} \\
&=\!\!\!\prod_{j\in \{s,m\}}\!\!\!\exp\bigg\{\!\!-2\pi\lambda_{j}\int_{r_{j}}^{\infty}\!\!\left( 1-\frac{1}{\left( 1+s p_{j} \|y\|^{-\al} \right)^{\Psi_{j}}} \right)r\mathrm{d}r\bigg\} \nn\\
&=\!\!\!\prod_{j\in \{s,m\}}\!\!\exp\bigg\{\!\!-2\pi\lambda_{j}\int_{r_{j}}^{\infty}\left( \frac{\displaystyle\sum_{i=1}^{\Psi_{j}}\binom{\Psi_{j}}{i}\left( s p_{j} r^{-\al} \right)^{i} }{\left( 1+s p_{j} r^{-\al} \right)^{\Psi_{j}}} \right)r\mathrm{d}r\bigg\}\nn \\
&=\prod_{j\in \{s,m\}}\exp\bigg\{\bigg.-\frac{2 \pi \lambda_{j} \left( sp_{j} \right)^{\frac{2}{a}}}{\al}\nn\\
 &~~~~~~~~~\times\!\!\bigg. \int_{w_{j}}^{1} \displaystyle \sum_{i=1}^{\Psi_{j}}\binom{\Psi_{j}}{i} t^{\Psi_{j}-1-i-\frac{2}{a}}\left( 1-t \right)^{i+\frac{2}{\al}-1}\mathrm{d}t\bigg\}\label{laplace5} \\
&=\prod_{j\in \{s,m\}}\exp\bigg\{\bigg.-\frac{2 \pi \lambda_{j} {\left( sp_{j} \right)}^{\frac{2}{a}}}{\al}\!\!\nn\\
&~~~~~~~~~\times \bigg.\sum_{i=1}^{\Psi_{j}}\!\!\binom{\Psi_{j}}{i}\mathrm{B}^{'}\!\left(\! \Psi_{j}\!+\!i\!-\!\frac{2}{a},i\!+\!\frac{2}{a}, w_{j}\right)\bigg\},
\end{align}
where~\eqref{laplace1} considers the independence among the locations of the BSs. Next,~\eqref{laplace2} results due to the independence between the spatial and the fading distributions, while in~\eqref{laplace3}, we have substituted the Laplace transform of $g_{y,j}$ following a Gamma distribution. We continue with the application of the property of the probability generating functional (PGFL)~\cite{Chiu2013a}, in order to obtain~\eqref{laplace4}. Moreover, in the following two equations, we convert the Cartesian coordinates to polar coordinates, and we apply the Binomial theorem. The calculation of the integral in~\eqref{laplace5} is obtained by means of many algebraic manipulations after making first the substitution $u=\left( s p_{j} \right)^{-\frac{1}{\al}} r $ and then $\left( 1+u^{-\al} \right)^{-1}\rightarrow t$. The last equation is obtained by using the definition of the incomplete beta function $B_{x}\left( p,q \right)$ defined in \cite[Eq.~(8.391)]{Gradshteyn2007}. Specifically, after defining $B_{x}^{'}\left( p,q \right)=\int_{x}^{1}t^{p-1}\left( 1-t \right)^{q-1}\mathrm{d}t$ as the complimentary incomplete Beta function, we obtain the desired result.

% , while in~\eqref{coverage_definition4}, we have considered that the term containoise and the aggregate interference from all BSs are independent.

% In other words, the Laplace transform of the sum equals the product of the Laplace transforms.

\section{Proof of Lemma~\ref{DerivativeOfLaplace}}\label{DerivativeOfLaplaceproof}
According to~\eqref{coverage_definition4}, the evaluation of the coverage probability demands the $k$th derivative of the Laplace transform of noise plus interference. By writting $\mathcal{L}_{\mathcal{I}N}\{s\left( \mathcal{I}+N \right)\}= \mathrm{e}^{-sN}\mathcal{L}_{\mathcal{I}}\{s \}$ as the composition $f\left( g\left( s \right) \right)$ with $f\left( x \right)$ and
\begin{align}
 g\left( s \right)&=2\pi \sum_{j\in \{s,m\}}^{2}\lambda_{j}\int_{w_{j}}^{\infty}\left( -1+\frac{1}{\left(1+s p_{j} \|y\|^{-\al}\right)^{\Psi_{j}}} \right)r\mathrm{d}r\nn\\
 &-sN.
 \end{align}
Application of Fa\`{a} di Bruno's formula provides the $k$th derivative. Specifically, Fa\`{a} di Bruno's formula is expressed as
\begin{align}
\frac{\mathrm{d}^{k}f\left( g\left( s \right) \right)}{\mathrm{d}s^{k}}&=\sum_{m=1}^{k}\sum_{\substack{i_{1}+i_{2}+\ldots+i_{k}=m\\
i_{1}+2i_{2}+\ldots+ik_{i}=k}
}f^{m}\left( g\left( x \right) \right)\frac{k!}{\prod_{j=1}^{k}j!^{i_{j}}i_{j}}\nn\\
&\times \prod_{j=1}^{k}\left( g^{\left( j \right)}\left( x \right) \right)^{i_{j}}.
\end{align}
Regarding the $j$th derivative of $g\left( s \right)$, we have
\begin{align}
 g^{j}\left( s \right)&=-N\mathbbm{1}_{j=1}+2\pi\sum_{i=1}^{2}\lambda_i\left( -1 \right)^{j}\frac{\left( \Psi_{i}+j-1 \right)~} {\left(\Psi_{i}-1 \right)!}\nn\\
 &\times \int_{w_{i}}^{\infty} \frac{r^{1-j\al}}{\left(1+s p_{i} \|y\|^{-\al}\right)^{\Psi_{i}+j}} \mathrm{d}r\nn\\
 &=-N\mathbbm{1}_{j=1}+2\pi\sum_{i=1}^{2}\left( -1 \right)^{j}D_{i}\left( j \right)p_{i}^{j}\left( s p_{i} \right)^{\frac{2}{\al}-j},
\end{align}
where the second equation is obtained similarly to the last two steps of the derivation of the Laplace transform in~Appendix~\ref{LaplaceTransformproof}.
\section{Proof of Theorem~\ref{AverageAchievableRate1Theorem}}\label{AverageAchievableRate1TheoremProof}
Following a standard procedure, the mean achievable rate of the typical user in the case of the $i$th event is obtained by means of its definition as
\begin{align}
 R_{i}&=\EE\left[\log_{2}\left( 1+\gamma_{i} \right) \right] \nn\\
 &=\frac{1}{\ln 2}\EE\left[\ln\left( 1+\gamma_{i} \right) \right]\nn\\
 &=\frac{1}{\ln 2}\int_{0}^{\infty}\mathbb{P}\left( \gamma_{i} >\mathrm{e}^{\mathcal{T}}-1\right)\mathrm{d}\mathcal{T}\nn\\
 &=\frac{1}{\ln 2}\int_{0}^{\infty}\mathbb{P}\left( \gamma_{i} >\theta\right)\frac{1}{1+\theta}\mathrm{d}\theta, 
\end{align}
where the last step concludes the proof by making a change of variables and applying the definition of the coverage probability.
Taking into account that the association events between the SBS and the MBS are mutually exclusive, the theorem is proved by means of the law of total probability. 
\end{appendices}

\bibliographystyle{IEEEtran}

\bibliography{mybib}
\end{document}